\documentclass[12pt,english,fleqn,liststotoc,bibtotoc,idxtotoc,BCOR7.5mm,tablecaptionabove]{amsart}
\usepackage[T1]{fontenc}
\usepackage[latin1]{inputenc}
\usepackage{geometry}
\usepackage{color}
\usepackage{amsthm}
\usepackage{amstext}
\usepackage{amssymb}
\usepackage{amsmath}
\usepackage{graphicx}
\usepackage{young}
\usepackage{xypic}
\usepackage{amscd}
\usepackage[normalem]{ulem}
\usepackage[section]{placeins}
\usepackage[svgnames]{xcolor}
\usepackage[unicode=true,
bookmarks=true,bookmarksnumbered=true,bookmarksopen=false,
breaklinks=false,pdfborder={0 0 1},backref=false,colorlinks=true]
{hyperref}
\hypersetup{pdftitle={},
	pdfauthor={Nick Early},
	pdfsubject={},
	pdfkeywords={},
	linkcolor=black, citecolor=black, urlcolor=blue, filecolor=blue,  pdfpagelayout=OneColumn, pdfnewwindow=true,  pdfstartview=XYZ, plainpages=false, pdfpagelabels}
\usepackage{breakurl}

\usepackage{scalerel}

\makeatletter

\theoremstyle{plain}
\newtheorem{thm}{\protect\theoremname}
\theoremstyle{plain}

\theoremstyle{plain}

\theoremstyle{remark}

\theoremstyle{plain}

\theoremstyle{plain}

\theoremstyle{remark}

\theoremstyle{remark}
\newtheorem{rem}[thm]{\protect\remarkname}
\theoremstyle{definition}
\newtheorem{defn}[thm]{\protect\definitionname}
\theoremstyle{definition}

\theoremstyle{plain}

\theoremstyle{plain}

\numberwithin{thm}{section}

\usepackage{ifpdf}
\ifpdf

\IfFileExists{lmodern.sty}
{\usepackage{lmodern}}{}

\fi 

\usepackage{multicol}

\newenvironment{nouppercase}{%
	\renewcommand{\uppercasenonmath}[1]{}}{}

\makeatother

\providecommand{\claimname}{\inputencoding{latin9}Claim}
\providecommand{\conjecturename}{\inputencoding{latin9}Conjecture}
\providecommand{\corollaryname}{\inputencoding{latin9}Corollary}
\providecommand{\definitionname}{\inputencoding{latin9}Definition}
\providecommand{\examplename}{\inputencoding{latin9}Example}
\providecommand{\lemmaname}{\inputencoding{latin9}Lemma}
\providecommand{\notename}{\inputencoding{latin9}Note}
\providecommand{\propositionname}{\inputencoding{latin9}Proposition}
\providecommand{\questionname}{\inputencoding{latin9}Question}
\providecommand{\remarkname}{\inputencoding{latin9}Remark}
\providecommand{\theoremname}{\inputencoding{latin9}Theorem}
\providecommand{\problemname}{\inputencoding{latin9}Problem}

\newcommand\twoheaduparrow{\mathrel{\rotatebox{90}{$\twoheaduparrow$}}}
\newcommand\twoheaddownarrow{\mathrel{\rotatebox{270}{$\twoheaddownarrow$}}}

\begin{document}
	\author{N\MakeLowercase{ick} E\MakeLowercase{arly}}
	\thanks{Max Planck Institute for Mathematics in the Sciences.  Email: \href{nick.early@mis.mpg.de}{nick.early@mis.mpg.de}}
	\title[Moduli Space Tilings and Lie-Theoretic Color Factors]{Moduli Space Tilings and Lie-Theoretic Color Factors}
	\begin{nouppercase}
		\maketitle
	\end{nouppercase}
	\begin{abstract}
		A detailed understanding of the moduli spaces $X(k,n)$ of $n$ points in projective $k-1$ space is essential to the investigation of generalized biadjoint scalar amplitudes, as discovered by Cachazo, Early, Guevara and Mizera (CEGM) in 2019. But in math, conventional wisdom says that it is completely hopeless due to the arbitrarily high complexity of realization spaces of oriented matroids.  In this paper, we nonetheless find a path forward.  
		
		We present a Lie-theoretic realization of color factors for color-dressed generalized biadjoint scalar amplitudes, formulated in terms of certain tilings of the real moduli space $X(k,n)$ and collections of logarithmic differential forms, resolving an important open question from recent work by Cachazo, Early and Zhang.  The main idea is to replace the realization space decomposition of $X(k,n)$ with a large class of overlapping tilings whose topologies are individually relatively simple.  So we obtain a collection of color-dressed amplitudes, each of which satisfies $U(1)$ decoupling separately.  The essential complexity reemerges when they are all superposed.
		
	\end{abstract}

	\begingroup
	\let\cleardoublepage\relax
	\let\clearpage\relax
	\tableofcontents
	\endgroup

\section{Introduction}

	In 2019, in \cite{CEGM2019} Cachazo, Early, Guevara and Mizera (CEGM) discovered a remarkable family of generalizations of the biadjoint scalar partial amplitudes which have an essentially recursive character with respect to soft and hard limits.  The defining formula is an integral over the moduli space $X(k,n)$ of $n$ points in projective space $\mathbb{P}^{k-1}$, localized to the solutions of the scattering equations.  A second formula, the Global Schwinger Parameterization, was introduced in \cite{CE2020B} which uses the parameterization of the positive configuration space to assemble the set of (tree-level) Schwinger parameterized Feynman diagrams into a single integral over $\mathbb{R}^{(k-1)(n-k-1)}$, see Figure \ref{fig:projectiontropgrass25}.  Recently the Global Schwinger Parameterization has been substantially generalized, see \cite{AFSPT2023}.
	
	Since \cite{CEGM2019}, much work has focused a special case, which is modeled on positive configuration space, or the torus quotient of the positive Grassmannian.  When $k=3$, the positive configuration space consists of projective equivalence classes of generic convex point configurations in $\mathbb{P}^2$.

Recently, Cachazo, Early and Zhang \cite{CEZ2022,CEZ2023} introduced an extension of the CEGM integral on $X(3,n)$ associated to non-convex point configurations in real projective space, developing a combinatorial framework for decoupling-type relations for color-dressed generalized biadjoint scalar amplitudes.  A Feynman diagram expansion was given by way of the Laplace transform of a new family of polyhedral fans, chirotopal tropical Grassmannians \cite{CEZ2022}, which are subsets of the tropical Grassmannian defined using ideas from oriented matroids \cite{Orientedmatroids}.

However, a Lie-theoretic realization of the color factors\footnote{See \cite{Dixon} and \cite{Mangano:1990by} for comprehensive reviews.} was missing: the generalized color factors of \cite{CEZ2023}, defined purely combinatorially, should involve certain traces of products of basis elements in adjoint representation of the Lie algebra $\mathfrak{su}(N)$.  This paper fills in that gap.  In our construction, there are analogs of important identities such as the vanishing of the color-dressed amplitude when the generator assigned to one of its particles is reassigned to a multiple of the identity matrix; this phenomenon is called $U(1)$-decoupling.  In our story such decoupling identities hold very naturally, as happens for $k=2$ biadjoint scalar amplitudes, where each color factor is the trace of a product of Lie algebra generators of $SU(N)$.  However, our color factors are rather different: they are products of certain collections of $k-1$ standard color factors.  

Let us introduce the main players briefly here, returning in later sections for discussion and motivation in more detail.

Fix a basis $(T^1,\ldots, T^{N^2-1})$ of the adjoint representation of $\mathfrak{su}(N)$ and let $T^{N^2}$ be proportional to the identity in $U(N)$, chosen such that $\text{tr}(T^iT^j) = \delta^{ij}$, following  common conventions.

	For any cyclic order $\sigma = (\sigma_1\cdots \sigma_n)$ on $\{1,\ldots, n\}$ (also called a \textit{color order}) and a tuple $(a_1,\ldots, a_n)$ of elements in $\{1,\ldots, N^2-1\}$, possibly with repetitions, denote by $c_\sigma$ the $(2,n)$ \textit{color factor} associated to $\sigma$, that is, the product of the trace of $n$ given generators in the cyclic order $\sigma$:
	\begin{eqnarray*}
		c_\sigma & = & \text{tr}(T^{a_{\sigma_1}}\cdots T^{a_{\sigma_n}}).
	\end{eqnarray*}
	This number is closely related to the structure coefficients of $\mathfrak{su}(N)$; see \cite{Dixon} for a general overview.  Assuming that $N$ is sufficiently large, by evaluating the color-dressed amplitude for different choices of the elements $T^{a_1},\ldots, T^{a_n}$ one can extract from the full color-dressed amplitude each partial amplitude, one for each color order.  For example, in the case of the biadjoint scalar, for any two tuples of elements $(T^{a_1},\ldots, T^{a_n})$ and $(T^{\tilde{a}_1},\ldots, T^{\tilde{a}_n})$ of $\mathfrak{su}(N)$, and  the color-dressed amplitude can be written as 
\begin{equation}\label{coBS}
	{\mathcal M}_n(\{ k_i,a_i,{\tilde a}_i\} ) =\!\!\! \sum_{\alpha,\beta \in S_{n}/{\mathbb{Z}_n}} \!\!\! {\text {tr}}\left( T^{a_{\alpha (1)}}T^{a_{\alpha (2)}} \cdots T^{a_{\alpha (n)}}\right)\!\! {\text{ tr}}\left( T^{{\tilde a}_{\beta (1)}}T^{{\tilde a}_{\beta (2)}} \cdots T^{{\tilde a}_{\beta (n)}}\right)m_n(\alpha,\beta). 
\end{equation}
	where the partial biadjoint scalar amplitudes $m_n(\alpha,\beta)$ are characterized by a pair of cyclic orders $\alpha,\beta$.  Here $k_1,\ldots, k_n$ are momentum vectors.

	Now suppose that $2\le k\le n-2$.  Fix a tuple $L = (\ell_1,\ldots, \ell_{k-1})$ of distinct elements of $\{1,\ldots, n\}$.  A type $(k,n)$ \textit{color order} is a $(k-1)$-tuple ${\overrightarrow{\sigma}} = (\sigma^1,\ldots,\sigma^{k-1})$ of cyclic orders $\sigma^1,\ldots,\sigma^{k-1}$ on respectively $\{1,\ldots, n\}\setminus (L \setminus \ell_j)$ for $j=1,\ldots, k-1$.  Denote by $\text{CO}^{(k)}_n$ the set of all type $(k,n)$ color orders.  Given a color order ${\overrightarrow{\sigma}} = (\sigma^1,\ldots, \sigma^{k-1}) \in \text{CO}^{(k)}_n$, the \textit{color factor} $c_{\overrightarrow{\sigma}}$ is the product of the $k-1$ color factors $c_{\sigma^j}$, namely
	\begin{eqnarray}\label{eq: color factor intro}
		c_{\overrightarrow{\sigma}} & = & \prod_{j=1}^{k-1}c_{\sigma^j}.
	\end{eqnarray}
		The permutation invariant color-dressed generalized biadjoint scalar amplitude\footnote{We urge caution here: Equation \eqref{eq: color factor intro} is different from the color-dressed generalized biadjoint amplitude introduced in \cite{CEZ2022}, wherein the color factors were formal variables!  We here adapt the definition to our realization of the color factors.  It will be an important problem to investigate how to pass from that definition to ours.} is given by 
	\begin{eqnarray}\label{eq: color dressed integrand generator perm invariant intro}
		\mathcal{M}^{(k)}_\mathbf{c} & = & \sum_{\overrightarrow{\alpha},\overrightarrow{\beta} \in \text{CO}^{(k)}_n}c_{\overrightarrow{\alpha}} \tilde{c}_{\overrightarrow{\beta}} m^{(k)}\left(\overrightarrow{\alpha},\overrightarrow{\beta}\right)
	\end{eqnarray}
	where the sum is over the set $\text{CO}^{(k)}_n$ of all pairs of type $(k,n)$-color orders $\overrightarrow{\sigma} = (\sigma^1,\ldots,\sigma^{k-1})$.
See Section \ref{sec: color-orders Xkn} for details, discussion and examples.
\begin{rem}
	Here each color factor $c_{\sigma^j}$ in Equation \eqref{eq: color factor intro} is the trace of a single tuple of generators of $\mathfrak{su}(N)$; however in Section \ref{sec: permutation invariance}, we tweak this condition and take each $c_{\sigma^j}$ to be the trace of a collection of its own set of elements in the cyclic order prescribed by $\sigma^j$.  In other words, in Equation \eqref{eq: color factor intro}, the color algebra is still the $\mathfrak{su}(N)$ color algebra, whereas in Section \ref{sec: permutation invariance} we generalize it to $\mathfrak{su}(N)^{\otimes (k-1)} = \mathfrak{su}(N)\otimes\cdots \otimes \mathfrak{su}(N)$.  The latter has many partial decouplings and in particular $n$ full $U(1)$ decouplings, so we may regard the color factors in Equation \eqref{eq: color factor intro} to be (an important) special case of it.
\end{rem}

	\vspace{.1in}

	Let $X(k,n) = G^\circ(k,n)\slash \left(\mathbb{R}_{\not=0}\right)^n$ be the torus quotient of the open real Grassmannian $G^\circ(k,n)$ where all maximal minors are nonzero; then when $k=2$ this is the moduli space $M_{0,n} = X(2,n)$.  For each $(k-1)$ tuple $L = (\ell_1,\ldots, \ell_{k-1})$ of distinct elements in  $\{1,\ldots, n\}$ we define an embedding $\varphi_L$ of $X(k,n)$ into the $(k-1)$-fold Cartesian product of $k-1$ copies of $X(2,n-(k-2))$,
	$$\varphi_L:X(k,n) \rightarrow X(2,n-(k-2)) \times \cdots \times X(2,n-(k-2))$$
	as follows.  Denote 
	$$\mathcal{H}_{j_1,\ldots, j_{d}} = H_{j_1}\cap \cdots \cap H_{j_{d}}$$
	for any subset $\{j_1,\ldots, j_{d}\}$ of $\{1,\ldots, n\}$.  We shall construct from each generic point in $X(k,n)$ a color order $\overrightarrow{\sigma} = (\sigma^1,\ldots, \sigma^{k-1})$.  For each $j=1,\ldots, k-1$, the intersection of $\mathcal{H}_{L\setminus \{\ell_j\}}$ with the remaining hyperplanes in turn lie on an embedded projective line in some (dihedral) order.  Choose a cyclic order for each, say $\overrightarrow{\sigma} = (\sigma^1,\ldots, \sigma^{k-1})$, where 
	\begin{eqnarray}
		\sigma^i & = & (\sigma^i_1,\ldots, \sigma^i_{n-(k-2)}).
	\end{eqnarray}
		\begin{figure}[h!]
	\centering
	\includegraphics[width=.3\linewidth]{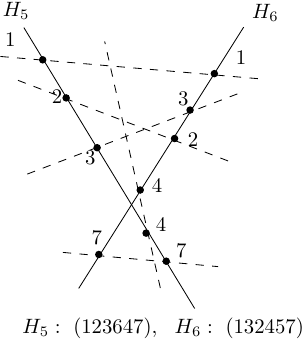}
	\caption{The (generically defined) projection $\varphi_{(5,6)}$ of $X(3,7)$ onto $X(2,6)\times X(2,6)$ determines an equivalence class of four distinct color orders $\overrightarrow{\sigma} = ((123647),(132457)),((746321),(132457)),\ldots$, since for the color order we distinguish a cyclic order from its reverse, whereas in $X(2,n)$ there is no such distinction.  The intersection $H_5\cap H_6$ is left unlabeled, though it is necessary to define the projection $\varphi_{(5,6)}$. We conversely obtain what we call informally a \textit{convolution} of two curvy $3$-dimensional associahedra: the generic arrangements of points on the lines $H_5,H_6$ determines uniquely the remaining five lines $H_1,H_2,H_3,H_4,H_7$, hence an open subset of $X(3,7)$.  The compactification of the image is interesting and we discuss in some detail later.}
	\label{fig:convolvedcolororders7}
\end{figure}
	See Figure \ref{fig:convolvedcolororders7} in the case that $L = (5,6)$.

	Explicit coordinates may be introduced:
	\begin{eqnarray}\label{eq:embedding}
		g & \mapsto & \left(\left(w^{(\widehat{\ell_1})}_{i_1,j_1}(g)\right),\ldots,\left( w^{(\widehat{\ell_{k-1}})}_{i_{k-1},j_{k-1}}(g)\right)\right)
	\end{eqnarray}
	where $\widehat{\ell_t} = L \setminus \{\ell_t\}$ and $\{i_t,j_t\}$ ranges over all cyclically nonconsecutive pairs in $\{1,\ldots, n\} \setminus (L\setminus \{\ell_t\})$.  Here the $w^{(\widehat{\ell_t})}_{i_t,j_t}(g)$ are certain cross-ratios, pull-backs of the so-called u-variables, cross-ratios $u_{i,j}$ on $M_{0,n-(k-2)}$.  To be concrete, whenever $i,j$ are not cyclically adjacent, the \textit{u-variable} $u_{i,j}$ is given by $u_{i,j} = \frac{p_{i,j+1}p_{i+1,j}}{p_{i,j}p_{i+1,j+1}}$ with $p_{i,j}$ the maximal $2\times 2$ minors on the Grassmannian $G(2,n-(k-2))$, which are pulled back to $X(k,n)$ as the cross-ratios of the form $w^{(M)}_{ij} = \frac{p_{M,i,j'}p_{M,i',j}}{p_{M,i,j}p_{M,i',j'}}$ where $i',j'$ are the cyclic successors of $i,j$ in the set $\{1,\ldots, n\} \setminus M$.  For the reader's convenience, code providing an implementation of the cross-ratios $w^{(M)}_{i,j}$ is given in Appendix \ref{sec:planar cross-ratios}.  More generally, in order to construct CEGM integrands for arbitrary color orders we will need to define u-variables for cyclic orders other than the standard $(1,2,\ldots, n)$, but for sake of simplicity we will not belabor this point and leave it up to the reader to permute labels as needed.
	
	By the preceding discussion, it follows that the image of each such embedding intersects generically each of the $\left(\frac{(n-(k-2)-1)!}{2}\right)^{k-1}$ Cartesian products of curvy associahedra, which, after pulling back, induce a tiling of $X(k,n)$.  Each tile in such a tiling is characterized by explicit inequalities on cross-ratios of the form $w^{(M)}_{ij}$.  Each tile is what we informally call a convolution product of $k-1$ curvy associahedra\footnote{That is, the connected components of the real moduli space $M_{0,n-(k-2)}$.}, as the (compactification of the) preimage is \textit{not} a Cartesian product as we will see shortly for $X(3,5)$.  The case when all curvy associahedra are defined by the same cyclic order, obtained by deleting from $(1,2,\ldots, n)$ the elements in $L\setminus \ell$ for each $\ell\in L$, has a direct connection to zonotopal tilings.  

	To illustrate the main construction, we turn to the embedding $\varphi_{(1,2)}:X(3,5) \rightarrow X(2,4) \times X(2,4)$ represented in coordinates by 
	\begin{eqnarray}\label{eq: example 35}
		g \mapsto \left((w^{(1)}_{24},w^{(1)}_{35}), (w^{(2)}_{14},w^{(2)}_{35})\right),
	\end{eqnarray}	
	where
	$$w^{(1)}_{24} = \frac{p_{125}p_{134}}{p_{124}p_{135}},\ w^{(1)}_{35} = \frac{p_{123}p_{145}}{p_{135}p_{124}},\ w^{(2)}_{14} = \frac{p_{125}p_{234}}{p_{124}p_{235}},\ w^{(2)}_{35} = \frac{p_{123}p_{245}}{p_{235}p_{124}}$$
	are regarded as systems of cross-ratios\footnote{Incidentally, the right hand side of Equation \eqref{eq: example 35} can be written in terms of planar cross-ratios $w_{abc}$, defined in \cite{Early2019PlanarBasis}, see also Appendix \ref{sec:planar cross-ratios}, as $\left((w_{124},w_{135}w_{235}),(w_{124}w_{134},w_{235})\right).$} on $M_{0,4}$, since 
	$$w^{(1)}_{24} + w^{(1)}_{35} = w^{(2)}_{14} + w^{(2)}_{35} = 1.$$

	In our construction, we further associate to the color order ${\overrightarrow{\sigma}} = ((2345),(1345))$ the logarithmic differential form 
	\begin{eqnarray*}
		\Omega_{\overrightarrow{\sigma}}(X(3,5)) & = & d\log\left(\frac{w^{(1)}_{24}}{w^{(1)}_{35}}\right)\wedge d\log\left(\frac{w^{(2)}_{14}}{w^{(2)}_{35}}\right)\\
		& = & d\log\left(\frac{p_{125}p_{134}}{p_{123}p_{145}}\right)\wedge d\log\left(\frac{p_{125}p_{234}}{p_{123}p_{245}}\right).
	\end{eqnarray*}
	
	Note that the four color orders
	$$((2345,1345)),((5432,1345)),((2345,5431)),((5432),(5431))$$
	all correspond to the same differential form and integrand (up to an overall sign).  Namely, one can check using the formula in Appendix \ref{sec: CEGM integrand calculation} that $\Omega_{\overrightarrow{\sigma}}(X(3,5))$ is proportional to
	$$\frac{p_{124}}{p_{123} p_{125} p_{134} p_{145} p_{234} p_{245}}.$$
	
	Moreover, by inserting the square of this new expression into the CEGM formula, instead of the usual 3-Parke-Taylor factor
	$$PT^{(3)}_5 = \frac{1}{p_{123}p_{234}p_{345}p_{451}p_{512}},$$
	we compute
	\begin{eqnarray}\label{eq: 35 special}
		m^{(3)}_5({\overrightarrow{\sigma}},{\overrightarrow{\sigma}}) & = &\frac{1}{\mathfrak{s}_{123} \mathfrak{s}_{145}}+\frac{1}{\mathfrak{s}_{123} \mathfrak{s}_{245}}+\frac{1}{\mathfrak{s}_{134} \mathfrak{s}_{245}}+\frac{1}{\mathfrak{s}_{125} \mathfrak{s}_{134}}+\frac{1}{\mathfrak{s}_{125} \mathfrak{s}_{234}}+\frac{1}{\mathfrak{s}_{145} \mathfrak{s}_{234}}.
	\end{eqnarray}
	The residues of the partial amplitude in Equation \eqref{eq: 35 special} clearly reflect the structure of a hexagon\footnote{Actually, this is equivalent to the union of two adjacent curvy pentagons in $M_{0,5}$.}, instead of the usual pentagon which is reflected in the residues of the biadjoint scalar partial amplitude
	\begin{eqnarray}
		m^{(2)}_5(\mathbb{I}_5,\mathbb{I}_5) = \frac{1}{s_{12}s_{34}} + \frac{1}{s_{23}s_{45}} + \frac{1}{s_{34}s_{15}} + \frac{1}{s_{45}s_{12}} + \frac{1}{s_{15}s_{23}},
	\end{eqnarray}
	where we following the standard CHY convention for the biadjoint scalar in setting $\mathbb{I}_n = (123\cdots n)$.

	In Equation \eqref{eq: 35 special}, the kinematic parameters $\mathfrak{s}_{abc}$ are new objects, indexed by $3$-element subsets.  In general they satisfy the $n$ independent momentum conservation relations
$$\sum_{\{a,b,c\} \ni j}\mathfrak{s}_{abc}=0$$
for each $j=1,\ldots, n$.  We refer to \cite{CEGM2019} for details on their construction, including a $k\ge 3$ generalization of the spinor helicity formalism.

	Incidentally, we point out that the differential form $\Omega_{((2345,1345))}$ can be expressed in terms of the so-called planar cross-ratios as 
	$$\Omega_{((2345),(1345))}(X(3,5)) = d\log\left(\frac{w_{124}}{w_{135}w_{235}}\right)\wedge d\log\left(\frac{w_{124}w_{134}}{w_{235}}\right).$$
	which makes manifest the identification the four planar poles out of the six poles of the amplitude computed in Equation \eqref{eq: 35 special}, namely
	$$w_{124} \Leftrightarrow \mathfrak{s}_{125},\ \ w_{134} \Leftrightarrow \mathfrak{s}_{234},\ \ w_{135} \Leftrightarrow \mathfrak{s}_{145},\ \ w_{235} \Leftrightarrow \mathfrak{s}_{123}.$$
	Here we follow the indexing convention for nonfrozen $k$-element subsets from \cite{Early2019PlanarBasis,Early19WeakSeparationMatroidSubdivision}.  Moreover, we see that the two new poles in Equation \eqref{eq: 35 special} which are not in the planar basis of kinematic invariants \cite{Early2019PlanarBasis} are $\mathfrak{s}_{134}$ and $\mathfrak{s}_{245}$, since the two sets $\{1,3,4\}$ and $\{2,4,5\}$ are not intervals with respect to the cyclic order $(1,2,3,4,5)$.
	
	This starts to reveal the novelty of the topology here: based on the amplitude and Newton polytope calculations, we are seeing that the fundamental building blocks include curvy hexagons, not just associahedra.
	
	Let us point out something intriguing: based on Newton polytope calculations in Section \ref{sec: examples}, the analogous construction for any $(n-2,n)$, taking the convolution product of $k-1$ one-dimensional associahedra for a fixed global cyclic order, appears to have the structure of a $(k-1)$-dimensional permutohedron!

	\subsection{Comments on Tree-Level Scattering Amplitudes via Tropical Geometry}
	Feynman's construction of tree-level scattering amplitudes in perturbative Quantum Field Theory is, at heart, both combinatorial and tropical in nature.  The standard procedure starts from the interaction term in a Lagrangian, which encodes a set of propagators, additional particle-specific data, and instructions on how to build Feynman diagrams and ultimately scattering amplitudes. The amplitude is then formulated by replacing each Feynman diagram with a rational function in the kinematic invariants, possibly multiplied with additional data (such as color factors), and then summing.  Feynman diagrams are (at tree-level) acyclic weighted directed graphs with $n$ leaves, where each edge is decorated with a momentum vector in Minkowski space.  Momentum conservation holds at each vertex of the graph.  The interaction term in a Lagrangian determines the Feynman \textit{rules}.  These rules entail a collection of inverse \textit{propagators}\footnote{Or, channels, especially for 4-point amplitudes, as in the $s,t,u$ channels, which are pairwise crossing.} $P_J^2 = (p_{j_1} + \cdots + p_{j_t})^2$, such that the amplitude which is being constructed has a simple pole (and only where) some collection of inverse propagators vanish: $P_J^2=0$.  They also entail how the propagators can interact with each other in the amplitude: there is a rule specifying which collections of poles are \textit{compatible}, or have non-crossing channels.  Two channels $P_{J_1}^2$ and $P_{J_2}^2$ are \textit{crossing} if and only if all four intersections are nonempty\footnote{In combinatorics, this is known to translate to the condition that the common refinement of two matroid subdivisions of the second hypersimplex is again a matroid subdivision \cite{Early19WeakSeparationMatroidSubdivision}, see also \cite{Early2019PlanarBasis}.}: $J_1\cap J_2,\ J_1\cap J_2^c,\ J_1^c \cap J_2,\ J_1^c \cap J_2^c$.
	
	For each maximal collection of compatible propagators, in particular each Feynman diagram, the Feynman rules associate a number, a certain \textit{numerator}, that encodes properties of the particles such as flavor and color.  In our context we consider biadjoint scalar amplitudes, in which case these numerators involve the structure constants of special unitary Lie algebras $\mathfrak{su}(N)$.  For large enough $N$, one can extract their coefficients, which are homogeneous rational functions of the momenta.   These rational functions are called \textit{partial amplitudes}; the color order imposes a certain notion of \textit{planarity}, in the sense that their Feynman diagrams can be embedded as planar graphs in a disk with $n$ marked points on the boundary in a fixed cyclic (color) order.  So, the possible propagators in a biadjoint cubic scalar partial amplitude are of the form $(P_{\sigma_{i+1}} + P_{\sigma_{i+1}} + \cdots + P_{\sigma_j})^2$ for a given cyclic order $(\sigma_1,\ldots, \sigma_n)$.  In physics, this cyclic order\footnote{However, the partial amplitude is actually invariant under the full dihedral group.} is sometimes called a color order.  

	\section{Review: Global Schwinger Parameterization, Generalized Color Orders and Integrands}
	
Generalized biadjoint scalar amplitudes, introduced by Cachazo, Early, Guevara and Mizera (CEGM) \cite{CEGM2019}, are homogeneous, degree $-(k-1)(n-k-1)$ rational functions on the kinematic space $\mathcal{K}_{k,n}$ that are intimately related to the tropical Grassmannian \cite{SpeyerStumfelsTropGrass}, and can be viewed as interpolations between Feynman diagram constructions of collections of standard QFT amplitudes \cite{Early2019PlanarBasis,Early19WeakSeparationMatroidSubdivision}.  In the CHY formalism, when $k=2$, modifying the Parke-Taylor factor integrand is known to provide complete information about amplitudes at tree level for many Quantum Field Theories, including Yang-Mills, NLSM, and special Galileon \cite{CachazoNLSM}.  Recently works \cite{CU2022} and \cite{AFSPT2023} are developing these theories using the Global Schwinger Parameterization \cite{CE2020B}.

CEGM amplitudes can be defined in several ways; we present the most general.  Following the discovery by CHY of the scattering equations formalism (corresponding to CEGM amplitudes in the case $k=2$), in \cite{CEZ2023}, Cachazo, Early and Zhang gave the following definition: for two differential forms on the subset of the real moduli space $X(k,n)$ of $n$ generically positioned points in projective space $\mathbb{P}^{k-1}$, as determined by the set of simplicial cells\footnote{And certain compatibility conditions on integrands which differ by a single simplex flip, see \cite{CEZ2023}.} in two arrangements of hyperplanes corresponding to a pair of realizable oriented uniform matroids $\chi_{L_1},\chi_{L_2}$, the CEGM amplitude is given by
\begin{eqnarray}\label{eq: mkn defn scattering}
	m^{(k)}_n(\mathcal{I}_L,\mathcal{I}_R) & = & \int_{X(k,n)}d\mu_{k,n} \mathcal{I}_L\mathcal{I}_R = \sum_{c \in \text{Crit}(\mathcal{S})}\frac{1}{\det'(\Phi)} \mathcal{I}_L\mathcal{I}_R\bigg\vert_c,
\end{eqnarray}
where $\det'(\Phi)$ denotes the reduced Hessian determinant \cite{CEGM2019} of the scattering potential $\mathcal{S}$. The measure $d\mu_{k,n}$ in Equation \eqref{eq: mkn defn scattering} is defined so as to localize the integral to the solutions of the scattering equation.

Here $\mathcal{S}$ is the critical points of the \textit{scattering potential} 
$$\mathcal{S} = \sum_{1\le j_1<\cdots <j_k\le n} \log(p_J) \mathfrak{s}_J,$$
where $p_J$ is a Plucker coordinate on the Grassmannian; thanks to momentum conservation we have that $\sum_{J: J \ni j}\mathfrak{s}_J = 0$ and the function $\mathcal{S}$ is well-defined on $X(k,n)$.  The denominator of each CEGM integrand \cite{CEZ2023} is a product of minors in bijection with the set of simplices in the hyperplane arrangement specified by a given realizable oriented uniform matroid \cite{Orientedmatroids}.  The numerator of the CEGM integrand is much more subtle to define and we refer the reader to \cite{CEZ2022,CEZ2023}.  

It is well-understood\footnote{See \cite[Section 2]{CEGM2019}, where the principle was used to derive the Feynman diagram expansion for $m^{(3)}\left(PT^{(3)}(\mathbb{I}_6),PT^{(3)}(\mathbb{I}_6)\right)$.}, by integrating products of Parke-Taylor factors $PT(\alpha)PT(\beta)$ for various cyclic orders $\alpha$ on $\{1,\ldots, n\}$ that when $k=2$, for biadjoint scalar partial amplitudes, $m^{(2)}_n(\alpha,\beta)$ is the (signed) Laplace transform of the \textit{intersection} of two copies of the positive tropical Grassmannian $\text{Trop}^+G(2,n)$ labeled by the cyclic (color) orders $\alpha,\beta$.  When $k\ge 3$, there are some important subtleties.  The main construction of \cite{CEGM2019}, see also \cite{BC2019}, used the following alternate definition: a CEGM partial amplitude $m^{(3)}(PT^{(k)}(\alpha),PT^{(k)}(\beta))$ is equal (at least, under certain assumptions on the kinematics) to the Laplace transform of the intersection of two permutations of the positive tropical tropical Grassmannian, but there is still the question of the overall sign.  For $k=2$ the overall sign is $(-1)^{w(\alpha,\beta)}$, where $w(\alpha,\beta)$ is the relative winding number between $\alpha$ and $\beta$.

To complete the picture, one also needs to prove that CEGM integral in Equation \eqref{eq: mkn defn scattering} evaluates \cite{CEZ2022,CEZ2023} to the same number as the Laplace transform of the intersection of two \textit{chirotopal} tropical Grassmannians.  For general pairs of chirotopal tropical Grassmannians, a combinatorial interpretation of the signs is rather subtle and it remains to give a simple characterization even for the first new CEGM amplitudes $m^{(3)}_n (\mathcal{I}_L,\mathcal{I}_R)$ with $n=6,7,8$, for pairs of integrands $\mathcal{I}_L,\mathcal{I}_R$.

As for the integrands, when for example when $(k,n) = (3,6)$, up to relabeling there are four of them, one for each isomorphism class of pseudoline arrangements as in Figure \ref{fig:36gcotypeiii} in the type III case.  The four integrands are
\begin{eqnarray}\label{eq: CEGM integrands}
	&&\frac{1}{p_{123} p_{126} p_{156} p_{234} p_{345} p_{456}},\ \frac{1}{p_{123} p_{126} p_{145} p_{234} p_{356} p_{456}},\ \frac{p_{236}}{p_{123} p_{126} p_{145} p_{234} p_{256} p_{346} p_{356}},\\
	&& \frac{p_{135} p_{146} p_{236} p_{245}-p_{136} p_{145} p_{235} p_{246}}{p_{123} p_{126} p_{134} p_{145} p_{156} p_{235} p_{245} p_{246} p_{346} p_{356}},\nonumber
\end{eqnarray}
corresponding to the four isomorphism classes of rank 3 oriented uniform matroids on six elements, that is, connected components of $X(3,6)$ up to relabeling.
\begin{figure}[h!]
	\centering
	\includegraphics[width=0.63\linewidth]{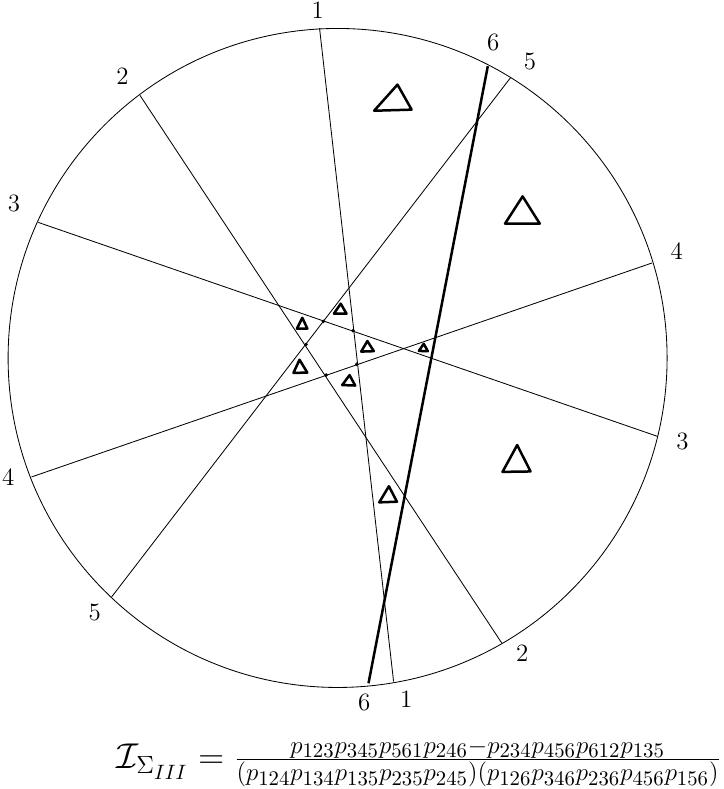}
	\caption{Integrand $\mathcal{I}_{\Sigma_{III}}$ from a type III line arrangement.  Denominator factors are in bijection with triangles in the arrangement.  The numerator is determined by identifying residues with neighboring integrands that differ by a triangle flip.  In other words, integrands for adjacent connected components of $X(3,6)$ should coincide on common codimension 1 boundaries.}
	\label{fig:36gcotypeiii}
\end{figure}

The first of these in Equation \ref{eq: CEGM integrands} is the $3$-Parke-Taylor factor, which is familiar from the study of the nonnegative Grassmannian; the others are new \cite{CEZ2023}.  See Figure \ref{fig:36gcotypeiii} for a representative type III line arrangement and its corresponding integrand.

As a special case, we abbreviate by $m^{(k)}_n$ the generalized biadjoint scalar amplitude arising when the product of the two Parke-Taylor factors $\mathcal{I}_{L} = \mathcal{I}_R = PT^{(k)}(1,2,\ldots, n) = 1\slash \prod_{j=1}^np_{j,j+1,\ldots, j+k-1}$ are inserted into the CEGM integral in Equation \eqref{eq: mkn defn scattering}; in this case, the resulting number can also be computed via the Global Schwinger Parameterization \cite{CE2020B}, using the parameterization of the \textit{positive} tropical Grassmannian $\text{Trop}^+G(k,n)$ \cite{SpeyerWilliams2003}, see also \cite{SpeyerWilliams2020,arkani2020positive}.  Here the Feynman diagram expansion is the Laplace transform of the positive tropical Grassmannian $\text{Trop}^+G(k,n)$, abbreviated $m^{(k)}_n$, given by 
\begin{eqnarray}\label{eq: mkn}
	m^{(k)}_n & = & \sum_{\mathcal{C}\in \text{Max Cones}(\text{Trop}^+G(k,n))}\frac{N(\mathcal{C})}{\prod_{X \in \text{Rays}(\mathcal{C})}X},
\end{eqnarray}
Here the numerators $N(\mathcal{C})$ are certain polynomials in (and the factors $X$ in the denominator are certain prescribed linear combinations of) the kinematic invariants, in duality with rays of $\text{Trop}^+G(k,n)$.  The numerator polynomials can in practice be computed by triangulating the cone generated by the rays $X$ and then summing the integral Laplace transforms of each simplicial cone in the triangulation.  Of course, closed formulas using the combinatorics of the rays of the tropical Grassmannian would be highly desirable!  The sum is over all (maximal) $(k-1)(n-k-1)$-dimensional cones $\mathcal{C}$ in the positive tropical Grassmannian and the product is over all rays $X$ of $\mathcal{C}$.  The Global Schwinger Parameterization, introduced in \cite{CE2020B}, tropicalizes the well-known positive parameterization of the nonnegative Grassmannian to glue together all Schwinger-parameterized Feynman diagrams into a single integral.  This parameterization has received a very interesting generalization in \cite{AFSPT2023} to higher loop order using the moduli space of Riemann surfaces.

\begin{figure}
	\centering
	\includegraphics[width=0.5\linewidth]{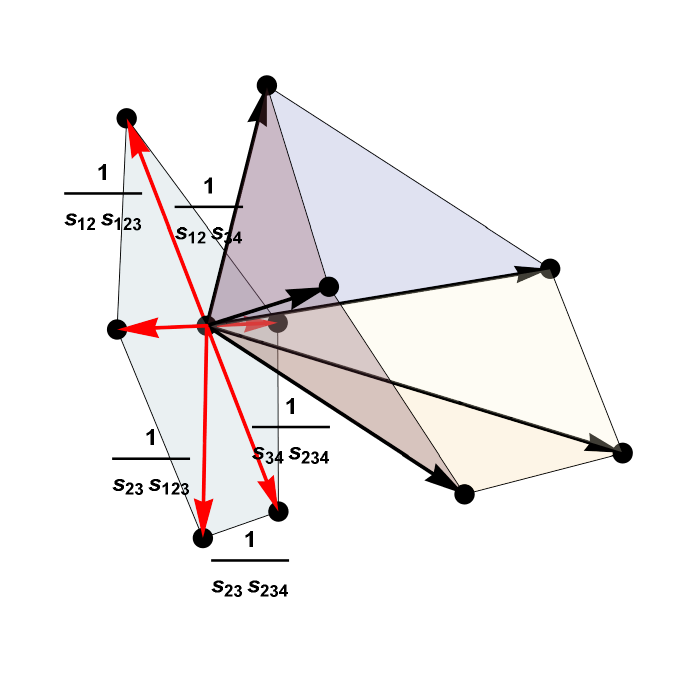}
	\caption{The Global Schwinger Parametrization \cite{CE2020B}, viewed as a projection of the positive tropical Grassmannian $\text{Trop}^+G(2,5)$. The five 2-dimensional cones correspond to individual Feynman diagrams. Rays (in red) are in duality with planar kinematic invariants $s_{i\cdots j}$.}
	\label{fig:projectiontropgrass25}
\end{figure}

From a physical point of view, the poles of $m^{(k)}_n$ have a quite novel structure.  For example, for $(k,n) = (3,6)$ there are three kinds of poles up to relabeling.  They are
$$\mathfrak{s}_{abc},\ \mathfrak{t}_{abcd} = \mathfrak{s}_{abc} + \mathfrak{s}_{abd} + \mathfrak{s}_{acd} + \mathfrak{s}_{bcd},\ \ R_{(ab,cd,ef)} = \mathfrak{t}_{abcd} + \mathfrak{s}_{cde} + \mathfrak{s}_{cdf}$$
for $\{a,b,c,d,e,f\} = \{1,2,3,4,5,6\}$.
It is important to note that the last pole is entirely new from the perspective of standard QFT; it is not obvious, but taking into account momentum conservation it has a cyclic symmetry 
$$R_{(ab,cd,ef)} = R_{(cd,ef,ab)}.$$
Nonetheless, as seen in the Generalized Feynman Diagram formalism \cite{BC2019,Early2019PlanarBasis,Early19WeakSeparationMatroidSubdivision} it is equivalent to six collections of standard QFT propagators, one for each of six 5-particle partial amplitudes.

From a combinatorial point of view this symmetry is not an accident; in fact it has a clear interpretation in terms of certain matroid subdivisions of an $5$-dimensional polytope called the hypersimplex $\Delta_{3,6}$.  Hypersimplices $\Delta_{k,n}$ are prevalent in combinatorics and geometry.

Continuing with $(k,n) = (3,6)$, the Feynman diagram construction consists of 48 compatible collections of metric trees, which leads to an expansion as a sum of 48 rational functions
$$m^{(3)}_6 = \frac{\mathfrak{t}_{1234}+\mathfrak{t}_{1256}+\mathfrak{t}_{3456}}{\mathfrak{t}_{1234} \mathfrak{t}_{3456} \mathfrak{t}_{1256} R_{12,34,56} R_{12,56,34}}+\frac{1}{\mathfrak{s}_{123} \mathfrak{s}_{456} \mathfrak{t}_{1234} \mathfrak{t}_{1456}} +\frac{1}{\mathfrak{s}_{156} \mathfrak{s}_{123} \mathfrak{t}_{1456} R_{16,45,23}}+ \text{ 45 more}.$$
The residue where any $R_{(ab,cd,ef)}=0$ is either zero or a product of \textit{three} copies of the four-particle biadjoint scalar partial amplitude.  It is nontrivial that there should be a Feynman diagram expansion \eqref{eq: mkn} which is numerically equal to Equation \eqref{eq: mkn defn scattering} when both integrands $\mathcal{I}_L,\mathcal{I}_R$ are the 3-Parke-Taylor factor $\prod_{j=1}^n p_{j,j+1,j+2}^{-1}$.  For a proof of the equality for this particular pair of CEGM integrands we refer the reader to \cite[Claim 1]{AHL2019Stringy}.  The general case for arbitrary pairs of integrands $\mathcal{I}_L,\mathcal{I}_R$ as constructed \cite{CEZ2023} from simplicial cells in pseudo-hyperplane arrangements, remains open.

We conclude this section with some brief remarks about momentum space, as formulated using the generalized spinor-helicity formalism introduced in \cite{CEGM2019}.

Momentum space for CEGM amplitudes consists of all $n$-tuples of rank $1$ matrices $\mathbb{K}_1,\ldots, \mathbb{K}_n$ of size $k\times k$ that satisfy \textit{momentum conservation}: the sum $\mathbb{K}_1+\cdots+\mathbb{K}_n$ is a matrix of corank $2$.    We write $\mathbb{K}_j = \lambda_j \otimes \tilde{\lambda}_j$ and $\lambda_j,\tilde{\lambda}_j\in \mathbb{C}^k$.

The $k$-analog of the Mandelstam invariant $s_{ij}$ is given by 
\begin{eqnarray*}
	\mathfrak{s}_{j_1\cdots j_k} & = & \det\left(\mathbb{K}_{j_1} + \cdots + \mathbb{K}_{j_k}\right).
\end{eqnarray*}
By the Cauchy-Binet identity, this may be rewritten as the product of two determinants, as
\begin{eqnarray*}
	\mathfrak{s}_{j_1\cdots j_k} = \langle j_1,\ldots, j_k\rangle \lbrack j_1,\ldots, j_k\rbrack,
\end{eqnarray*}
where $\langle j_1,\ldots, j_k\rangle = \det(\lambda_{j_1},\ldots, \lambda_{j_k})$ and $\lbrack j_1,\ldots, j_k\rbrack = \det(\tilde{\lambda}_{j_1},\ldots, \tilde{\lambda}_{j_k})$ are coordinates on the left and right factors of the twisted Cartesian product of a pair of Grassmannians, 
$$\left(G(k,n) \times G(k,n)\right)\slash (g_1,g_2)\sim (g_1 t_j,g_2 t_j^{-1}),$$
consisting of equivalence classes modulo the twisted diagonal action of the torus $\left(\mathbb{C}^\ast\right)^n$ given above.

Momentum conservation among the $n$ momentum matrices implies the $n$ independent equations $\sum_{J: J \ni j}\mathfrak{s}_J = 0$ for $j=1,\ldots, n$.

Poles $X=0$ of $m^{(k)}_n$ are given by $X=\sum_{J}\pi_J \mathfrak{s}_J$ for any ray $\pi = (\pi_J) \in \mathbb{R}^{\binom{n}{k}}$ of the (positive) tropical Grassmannian\footnote{And more generally, any ray of a given chirotopal tropical Grassmannian, in light of \cite{CEZ2022,CEZ2023}.}.

	\section{Color-Orders, Color Factors and CEGM Integrands}\label{sec: color-orders Xkn}
	
	Fix a tuple $L = (\ell_1,\ldots, \ell_{k-1})$ of distinct elements of $\{1,\ldots, n\}$.  A $(k,n)$ \textit{color order} is a $(k-1)$-tuple ${\overrightarrow{\sigma}} = (\sigma^1,\ldots,\sigma^{k-1})$ of cyclic orders $\sigma^1,\ldots,\sigma^{k-1}$ on respectively $\{1,\ldots, n\}\setminus (L \setminus \ell_j)$ for $j=1,\ldots, k-1$.

\begin{defn}\label{def: color factor Xkn}
	Given a $(k,n)$ color order ${\overrightarrow{\sigma}} = (\sigma^1,\ldots, \sigma^{k-1})$, the \textit{color factor} $c_{\overrightarrow{\sigma}}$ is the product of the $k-1$ color factors $c_{\sigma^j}$, namely
	\begin{eqnarray}
		c_{\overrightarrow{\sigma}} & = & \prod_{j=1}^{k-1}c_{\sigma^j} 
	\end{eqnarray}
\end{defn}

	Given a tuple $L$ as above, let ${\overrightarrow{\sigma}}$ be a $(k,n)$-color order.  We define a logarithmic differential form $\Omega_{\overrightarrow{\sigma}}$ on $X(k,n)$ by 
	$$\Omega_{{\overrightarrow{\sigma}}} =\Omega^{L\setminus \{\ell_{1}\}}(M_{0,n-(k-2)}(\sigma^{1}))\wedge\Omega^{L\setminus \{\ell_2\}}(M_{0,n-(k-2)}(\sigma^{2}))\wedge \cdots \wedge \Omega^{L\setminus \{\ell_{k-1}\}}(M_{0,n-(k-2)}(\sigma^{k-1})),$$
	where $\Omega^{L\setminus \{\ell_j\}}(M_{0,n-(k-2)}(\sigma^j))$ is the canonical form of the open connected component $M_{0,n-(k-2)}(\sigma^j)$ of $M_{0,n-(k-2)}$ defined by the color order $\sigma^j$, pulled back to $X(k,n)$ along the projective line
	$$H_{\ell_1} \cap \cdots \cap \widehat{H_{\ell_j}}\cap \cdots \cap H_{\ell_{k-1}},$$
	in $\mathbb{P}^{k-1}$.  This is compatible with the embedding of Equation \eqref{eq:embedding}.  Here the arrangement of hyperplanes $H_1,\ldots, H_n$ in $\mathbb{P}^{k-1}$ represents a (generic) point in $X(k,n)$.  A general theory of canonical forms has been developed recently in \cite{ABL,ABHY2018,AHL2019Stringy,arkani2020positive}.  Thus, the logarithmic forms $\Omega_{\overrightarrow{\sigma}}$ are constructed from canonical forms on $X(2,n-(k-2))$.  We define these now.  
		
		For the standard cyclic order $\sigma=(1,2,\ldots, n)$, the canonical form on the positive part $M^+_{0,n}$ of $M_{0,n}$ can be written conveniently as 
	\begin{eqnarray*}
		\Omega(M_{0,n}(\sigma)) & = & d\log\left(\frac{u_{13}}{1-u_{13}}\right) \wedge d\log\left(\frac{u_{14}u_{24}}{1-u_{14}u_{24}}\right)\wedge \cdots \wedge d\log\left(\frac{u_{1,n-1}u_{2,n-1}\cdots u_{n-2,n-1}}{1-u_{1,n-1}u_{2,n-1}\cdots u_{n-2,n-1}}\right),
	\end{eqnarray*}
	where the so-called u-variables $u_{ij} = \frac{p_{i,j+1}p_{i+1,j}}{p_{i,j}p_{i+1,j+1}}$ were first defined by Koba-Nielsen in \cite{KobaNielsen} in the context of string integrals.  We remind that the positive part $M^+_{0,n}$ is characterized by the inequalities $0<u_{ij}<1$.  In terms of Plucker variables, one has, more explicitly,
	\begin{eqnarray}
		\Omega(M_{0,n}(1,2,\ldots, n)) & = & d\log\left(\frac{p_{14}p_{23}}{p_{12}p_{34}}\right)\wedge d\log\left(\frac{p_{15}p_{34}}{p_{13}p_{45}}\right) \wedge \cdots \wedge d\log\left(\frac{p_{1,n}p_{n-2,n-1}}{p_{1,n-2}p_{n-1,n}}\right).
	\end{eqnarray}
	The form $\Omega(M_{0,n}(\sigma))$ is known to be proportional to the Parke-Taylor factor 
	$$PT(1,2,\ldots, n) = \frac{1}{p_{12}p_{23}\cdots p_{n1}}.$$
	Canonical forms (see \cite{ABL} for details on canonical forms and positive geometry more generally) for the other connected components of $M_{0,n}$ can be obtained by relabeling the indices on the Plucker coordinates in $\Omega(M_{0,n}(\sigma))$ accordingly, according to the cyclic order $\sigma$.

	For parameters $c_\sigma$, the \textit{color-dressed canonical form} $	\Omega_c(M_{0,n})$ on $M_{0,n}$ is defined by 
	\begin{eqnarray}\label{eq: color-dressed form M0n}
	\Omega_c(M_{0,n}) & = & \sum_{\alpha}\text{sgn}_n(\alpha_1,\ldots, \alpha_{n-1}) c_{(\alpha_1\cdots \alpha_{n-1}n)} \Omega(M_{0,n}(\alpha_1\cdots \alpha_{n-1}n)),
\end{eqnarray}
	where the sum is over all permutations $\alpha$ of $\{1,\ldots, n-1\}$.  The coefficient $\text{sgn}_n(\alpha_1,\ldots, \alpha_{n-1})$ is the sign of the permutation (in line notation) $(\alpha_1,\ldots, \alpha_{n-1})$; it implies that the $n$ $U(1)$ decoupling identities among forms are satisfied, when the color factor $c_{(\alpha_1\cdots \alpha_{n-1}n)}$ is replaced with  $c_{(\alpha_1\cdots \hat{j} \cdots \alpha_{n-1}n)}$ for any $j\in \{1,\ldots, n\}$.  Here $(\alpha_1\cdots \alpha_{n-1}n)$ denotes the prescribed cyclic order.
	
	Here we remind that the $c_\sigma$ are \textit{color factors}, that is products of traces of given collection of generators 
	$$T^{a_1},\ldots, T^{a_n} \in \mathfrak{su}(N)$$
	of the special unitary group $SU(N)$ for some large $N$, where we take $a_1,\ldots, a_n\in \{1,\ldots, N^2-1\}$ to index (a subset of) the elements in the standard basis of $\mathfrak{su}(N)$.  
	\begin{defn}
		By a \textit{$U(1)$ decoupling identity} we mean a minimal set of partial amplitudes which sums to zero when one of the generators $T^{a_j}$ is replaced by a multiple of the identity in $U(N)$.
	\end{defn}
	So, $U(1)$-decoupling amounts to identities\footnote{In matroid theory, such a set of Parke-Taylor factors would be called a \textit{circuit}, since the $n-1$ summands satisfy the unique given linear relation.} of the form
	$$PT(1,2,\ldots, n-1,n)+PT(1,2,\ldots, n,n-1) + \cdots + PT(1,n,2,\ldots, n-2,n-1)=0$$
	among Parke-Taylor factors.
	
	We now move to give the Lie-theoretic realization of color factors needed to finish the definition of color-dressed CEGM amplitudes.  This fills in an important missing piece from recent work with Cachazo and Zhang \cite{CEZ2022,CEZ2023}: in that work, it was possible to produce sets of integrands which were to satisfy a single identity with all $\pm1$ coefficients, but there was no general rule for determining the signs.  We formulate the notion of color-dressed logarithmic differential forms; color-dressed amplitudes are obtained from them via the CEGM formula, by applying the formula in Appendix \ref{sec: CEGM integrand calculation} to compute the corresponding integrand.
	
	Our notion of color factors builds on recent work of Cachazo, Early and Zhang \cite{CEZ2022}, which extended the construction of CEGM amplitudes \cite{CEGM2019} from integrands defined for convex point configurations (i.e., the so-called positive configuration spaces \cite{arkani2020positive}, where the integrand is the $k$-Parke-Taylor factor) to integrands associated to nonconvex point configurations representing chirotopal configuration spaces $X^\chi(k,n)$ for (at least in principle) any realizable oriented uniform matroid $\chi \in \{\pm1\}^{\binom{n}{k}}$.  Amplitudes were computed explicitly for all values of $(k,n)$ possible with existing technology, for $(3,6),(3,7),(3,8)$, with partial results for $(3,9)$ and $(4,8)$.  In \cite{CEZ2023}, algorithms were developed to generate irreducible decoupling identities for $X(3,n)$.  However, an explicit Lie-theoretic realization of the color factors themselves was missing.  
	
	In what follows, we give that realization.  More precisely, we introduce the notion of $(k,n)$-\textit{color factors} for any $k$, which are certain degree $k-1$ products of standard color factors, that is, where each standard color factor is the trace of a product of given generators $T^{a_1},\ldots, T^{a_n} \in \mathfrak{su}(N)$ of the special unitary group $SU(N)$.  It is important that $U(1)$ decoupling identities are automatically satisfied in our construction.

	\vspace{.1in}
	
	We fix a collection of generators $T^{a_1},\ldots, T^{a_n}\in \mathfrak{su}(N)$ of the special unitary group $SU(N)$, where we assume $N\gg n$ for sake of genericity.
	
	Given a $(k-1)$-tuple $L = (\ell_1,\ldots, \ell_{k-1})$ of distinct elements of $\{1,\ldots, n\}$, for each color order ${\overrightarrow{\sigma}}$, where ${\overrightarrow{\sigma}} =(\sigma^1,\ldots, \sigma^{k-1})$ consists of cyclic orders 		
$$\sigma^1 = \left(\sigma^{1}_1,\ldots, \sigma^1_{n-(k-2)}\right),\ \ldots,\ \sigma^{k-1} = \left(\sigma^{k-1}_1,\ldots, \sigma^{k-1}_{n-(k-2)}\right)$$
on respectively 
$$\left\{1,\ldots, n\right\}\setminus (L\setminus \{\ell_1\}),\ \ldots,\ \left\{1,\ldots, n\right\}\setminus (L\setminus \{\ell_{k-1}\}),$$ the color factor $c_{\overrightarrow{\sigma}}$ is a degree $k-1$ product of traces,
\begin{eqnarray}\label{eq: color factor}
	c_{\overrightarrow{\sigma}} & = & \text{tr}\left(T^{a_{\sigma^1_1}}T^{a_{\sigma^1_2}}\cdots T^{a_{\sigma^1_{n-(k-2)}}}\right)\cdots \text{tr}\left(T^{a_{\sigma^{k-1}_1}}T^{a_{\sigma^{k-1}_2}}\cdots T^{a_{\sigma^{k-1}_{n-(k-2)}}}\right)
\end{eqnarray}
	
It is easier to state for small $(k,n)$, where the notation is less cumbersome.

	For example, for $(k,n) = (3,5)$, if $L = (1,2)$ and ${\overrightarrow{\sigma}} = \{(2345),(1345)\}$, then 
	$$c_{{\overrightarrow{\sigma}}} = \text{tr}(T^{a_2}T^{a_3}T^{a_4}T^{a_5}) \text{tr}(T^{a_1}T^{a_3}T^{a_4}T^{a_5}).$$
	Incidentally, for readers familiar with the construction of \cite{CEZ2022,CEZ2023}, the color order ${\overrightarrow{\sigma}}$ is the projection onto the first two slots of two different Generalized Color Orders	(GCOs), namely
	$$((2345),(1345),(1245),(1235),(1234)),\text{ and } ((2543),(1543),(1254),(1253),(1243))$$
	which are descendants of the global cyclic color orders $(12345)$ and $(12543)$.  A GCO $\overrightarrow{\sigma} \in \text{CO}^{(3)}_n$ is said to be \textit{descendant} if it is obtained by deleting the labels $1,2,\ldots, n$ in turn from a single, global cyclic order.  

	Generalized Color Orders axiomatize the combinatorial data contained in the set of projections $X(k,n) \rightarrow X(k,k+2)$, noting that by duality we have $X(k,k+2) \simeq X(2,k+2) = M_{0,k+2}$.
	\begin{defn}[\cite{CEZ2022}]
		Fix an arrangement of hyperplanes $H_1,\ldots, H_N$ in generic position in the real projective space $\mathbb{P}^{k-1}$.  A $(k,n)$ Generalized Color Order (GCO) $\Sigma$ is the collection of dihedral orders, induced on the projective lines defined by taking all $(k-2)$-fold intersections of the hyperplanes $H_1,\ldots, H_n$.
	\end{defn}
	Here let us point out that our present definition of \textit{color orders} (without the word generalized) consist of $(k-1)$-element collections of \textit{cyclic} orders on $n-(k-2)$ elements, which by construction can be identified with sub-collections of GCOs in the sense of \cite{CEZ2022}.  However, we note that a GCO is a collection of $\binom{n}{k-2}$ dihedral orders on $n-(k-2)$ elements, whereas our color orders are certain $(k-1)$-element collections of cyclic orders on $n-(k-2)$ elements.
	
	Continuing with our example, again for $(k,n) = (3,5)$, if $L = (2,1)$ and ${\overrightarrow{\sigma}} = ((2345),(1354))$, then similarly
	$$c_{{\overrightarrow{\sigma}}} = \text{tr}(T^{a_2}T^{a_3}T^{a_4}T^{a_5})\text{tr}(T^{a_1}T^{a_3}T^{a_5}T^{a_4}) ,$$
	but now the color order ${\overrightarrow{\sigma}}$ is the projection onto the first two slots of exactly one GCO, 
	$$((2345),(1354),(1254),(1325),(1324)),$$
	which descends from the unique global color order $(13254)$.
	
	Finally, for $(k,n) = (4,7)$, if for example $L = (1,2,3)$ and ${\overrightarrow{\sigma}} = \{(14756),(25476),(36547)\}$, then 
	$$c_{{\overrightarrow{\sigma}}} = \text{tr}\left(T^{a_1}T^{a_4}T^{a_7}T^{a_5}T^{a_6}\right)\text{tr}\left(T^{a_2}T^{a_5}T^{a_4}T^{a_7}T^{a_6}\right)\text{tr}\left(T^{a_3}T^{a_6}T^{a_5}T^{a_4}T^{a_7}\right).$$
	We invite the reader to compute the number of GCOs that project to ${\overrightarrow{\sigma}}$, using (for example) the ancillary data provided in \cite{CEZ2022}.
	
	We now come to our first main construction.  In what follows, the forms $\Omega_{\overrightarrow{\sigma}}$ are grouped together according to the tuple $L = (\ell_1,\ldots, \ell_{k-1})$.  This means that for each fixed $L$, letting ${\overrightarrow{\sigma}}$ vary over all collections of cyclic orders compatible with $L$, then the set of all such forms $\Omega_{\overrightarrow{\sigma}}$'s satisfy their own set of standard shuffle identities among Parke-Taylor factors, such as, in particular, U(1) decoupling.

	\begin{defn}
		Given a $(k-1)$-tuple $L =(\ell_1,\ldots, \ell_{k-1})$ of distinct elements $\ell_j\in \{1,\ldots, n\}$, for any $n$-tuple $(T^{a_1},\ldots, T^{a_n})$, the \textit{color-dressed form} $\Omega^{L}_\mathbf{c}(X(k,n))$ is the wedge product of $k-1$ pulled back color-dressed forms,
	\begin{eqnarray}\label{eq: color dressed form generator}
	\Omega^{L}_\mathbf{c}(X(k,n)) & = & \Omega^{L\setminus\{\ell_{1}\}}_{c}\wedge \Omega^{L\setminus\{\ell_{2}\} }_{c}\wedge \cdots \wedge \Omega^{L\setminus\{\ell_{k-1}\}}_{c}.
\end{eqnarray}
\end{defn}

	\begin{thm}\label{thm: U(1) decoupling kn}
		For each fixed $(k-1)$-tuple $L$, the form $\Omega^{L}_c(X(k,n))$ satisfies the $U(1)$-decoupling identity: when one of the generators $T^{a_j}$ is replaced with a scalar multiple of the identity in $U(N)$, then $\Omega^{L}_c(X(k,n))$ vanishes identically. 
	\end{thm}
	
	\begin{proof}

		Suppose we want to decouple the $n^\text{th}$ particle, that is we set $T^{a_n}$ to be any multiple of the identity in $U(N)$, say $T^{a_n} = \lambda\text{id}_{N\times N}$ for $\vert \lambda \vert = 1$, so in particular $T^{a_n}$ now commutes with all of $U(N)$.  There are two cases.  
		
		If $n \in L$, then exactly one of the $k-1$ factors $\Omega^{(L\setminus\{n\})}_{c}$ of 
		$$ \Omega^{L\setminus\{\ell_{1}\}}_{c}\wedge \Omega^{L\setminus\{\ell_{2}\}}_{c}\wedge \cdots \wedge \Omega^{L\setminus\{\ell_{k-1}\}}_{c}$$
		is modified, and after setting $T^{a_n} = \lambda\text{id}_{N\times N}$ then we have a standard $U(1)$ decoupling identity since now $\Omega^{L\setminus\{n\}}_{c}\equiv0$, while $\Omega^{L\setminus\{\ell\}}_{c}\not\equiv0$ for all $\ell\not=n$.  On the other hand, if $n \not\in L$ then \textit{all} $k-1$ factors vanish identically, completing the proof.
		
	\end{proof}

		As a consequence of Theorem \ref{thm: U(1) decoupling kn}, \textit{any} linear combination of the forms $\Omega^{(L)}_c(X(k,n))$, as $L$ varies, also satisfies the $n$ $U(1)$ decoupling identities, which opens the door to the possibility of additional constraints.  For example, one can construct a color-dressed CEGM integrand (using the formula in Appendix \ref{sec: CEGM integrand calculation}) which is left invariant by relabeling.  See Section \ref{sec: permutation invariance}, where we initiate the study of permutation invariant color-dressed integrands and amplitudes, for which we have to enlarge the color algebra, that is to replace the adjoint representation of $\mathfrak{su}(N)$ with its $(k-1)^\text{st}$ tensor power $\mathfrak{su}(N)^{\otimes(k-1)}$.

	\section{Examples: Forms, Integrands, GCOs and Newton Polytopes}\label{sec: examples}
	
	With $L = (2,1)$, the differential form defined by the color order
	$$\overrightarrow{\sigma} = ((23456),(13456)),$$
	is 
	\begin{eqnarray*}
		\Omega_{((23456),(13456))} & = & d\log\left(\frac{w^{(1)}_{24}}{w^{(1)}_{35}w^{(1)}_{36}}\right) \wedge d\log\left(\frac{w^{(1)}_{25}w^{(1)}_{35}}{w^{(1)}_{46}}\right) \wedge d\log\left(\frac{w^{(2)}_{35}}{w^{(2)}_{14}w^{(2)}_{46}}\right) \wedge d\log\left(\frac{w^{(2)}_{36}w^{(2)}_{46}}{w^{(2)}_{15}}\right)\\
		& = & d\log\left(\frac{w_{124}}{w_{135}w_{235}w_{136}w_{236}}\right) \wedge d\log\left(\frac{w_{125}w_{135}w_{235}}{w_{146}w_{246}w_{346}}\right) \\
		& \wedge & d\log\left(\frac{w_{235}}{w_{124}w_{134}w_{246}w_{346}}\right) \wedge d\log\left(\frac{w_{236}w_{246}w_{346}}{w_{125}w_{135}w_{145}}\right).
	\end{eqnarray*}
	In this representation, 11 distinct planar cross-ratios $w_{abc}$ are present, see Appendix \ref{sec:planar cross-ratios} for the definition.
	
	Here we list some integrands for $n=5,6,7$.  In particular, 
	$$\Omega_{((2345),(1345))}, \Omega_{((23456),(13456))},\Omega_{((234567),(134567))}$$
	are proportional to respectively
	$$\frac{p_{124}}{p_{123} p_{125} p_{134} p_{145} p_{234} p_{245}},\frac{p_{124} p_{125}}{p_{123} p_{126} p_{134} p_{145} p_{156} p_{234} p_{245} p_{256}},\frac{p_{124} p_{125} p_{126}}{p_{123} p_{127} p_{134} p_{145} p_{156} p_{167} p_{234} p_{245} p_{256} p_{267}},\ldots $$
	
	Let us now compute some Newton polytopes.  Consider the parameterization\footnote{There are obvious analogs of the parameterization for other $(k,n)$ color orders, that is, other $(k-1)$-tuples of cyclic orders.}
	$$M(x_{i,j}) = \left(
	\begin{array}{ccccccc}
		1 & 0 & 0 & x_{1,1} & x_{1,1}+x_{1,2} & x_{1,1}+x_{1,2}+x_{1,3} &\\
		0 & 1 & 0 & -x_{2,1} & -x_{2,1}-x_{2,2} & -x_{2,1}-x_{2,2}-x_{2,3} & \cdots \\
		0 & 0 & 1 & 1 & 1 & 1 & \\
	\end{array}
	\right).$$
	When $n=5$, keeping the first five columns of $M(x_{i,j})$, the Newton polytope of the product of all maximal minors is a hexagon (see also Equation \eqref{eq: 35 special} for the amplitude) and the form is proportional to 
	$$\frac{p_{124}}{p_{123} p_{125} p_{134} p_{145} p_{234} p_{245}}.$$
	When $n=6$, keeping the first six columns, we find that the Newton polytope of the product of all maximal minors is a four-dimensional polytope with 23 facets and 80 vertices.  The full f-vector is $(80, 162, 105, 23, 1)$.  We remark that even though the polynomials are \textit{not} subtraction free, at least for $n=5,6$ where we have checked directly, still it turns out that the Newton polytope is compatible with the CEGM amplitude obtained from
	\begin{eqnarray}\label{eq: eightGCOs Integrand}
		\mathcal{I}_{((23456),(13456))} = \frac{p_{124} p_{125}}{p_{123} p_{126} p_{134} p_{145} p_{156} p_{234} p_{245} p_{256}},
	\end{eqnarray}
	in the sense that the (inward) normal fan of the Newton polytope of the product of all maximal minors of $M(x_{i,j})$ provides the Global Schwinger Parameterization for the CEGM amplitude with the square of  $\mathcal{I}_{((23456),(13456))}$ as its integrand.
	
	There is another expression for Equation \eqref{eq: eightGCOs Integrand} as a linear combination of the CEGM integrands obtained in \cite{CEZ2023} labeled by the following eight GCOs, all of which have in common the first two entries:
	\begin{eqnarray*}
		\begin{array}{c}
			(23456,13456,12456,12356,12346,12345) \\
			(23456,13456,12654,12653,12643,12543) \\
			(23456,13456,12456,12365,12364,12354) \\
			(23456,13456,12546,12536,12436,12345) \\
			(23456,13456,12645,12635,12634,12543) \\
			(23456,13456,12654,12563,12463,12453) \\
			(23456,13456,12465,12365,12634,12534) \\
			(23456,13456,12564,12563,12436,12435) \\
		\end{array}
	\end{eqnarray*}
	The corresponding eight integrands are
	\begin{eqnarray*}
		& & \frac{1}{p_{123} p_{126} p_{156} p_{234} p_{345} p_{456}},\frac{1}{p_{123} p_{126} p_{134} p_{256} p_{345} p_{456}},\frac{1}{p_{123} p_{126} p_{145} p_{234} p_{356} p_{456}},\frac{1}{p_{123} p_{126} p_{156} p_{245} p_{345} p_{346}}\\
		& & \frac{1}{p_{123} p_{126} p_{145} p_{256} p_{345} p_{346}},\frac{1}{p_{123} p_{126} p_{134} p_{245} p_{356} p_{456}},\frac{p_{236}}{p_{123} p_{126} p_{145} p_{234} p_{256} p_{346} p_{356}},\\
		& & \frac{p_{136}}{p_{123} p_{126} p_{134} p_{156} p_{245} p_{346} p_{356}}.
	\end{eqnarray*}
	We leave it as an exercise to find the unique eight coefficients $\pm1$ that express Equation \eqref{eq: eightGCOs Integrand} in terms of these eight integrands.
	
	Let us now record an interesting connection to zonotopal tilings.
	
	\begin{rem}
		There is a well-known relationship between oriented matroids and zonotopal tilings, so it is natural to expect a connection in our story as well; indeed, our computations support this expectation.  For $k=3$ and $n=5,6,7,8$ we get respectively $2,8,62,908$ Generalized Color Orders that contain the standard cyclic orders $(234\cdots n)$ and $(134\cdots n)$ on their first two slots.  See O.E.I.S entry A006245.  In other words, we are enumerating the tiles in the convolution of two curvy associahedra that are defined by the pair of cyclic orders $((23\cdots n),(134\cdots n))$. The enumeration we obtain agrees with the number of zonotopal (rhombus) tilings of a centrally symmetric $2(n-2)$-gon!  Indeed, this correspondence can be made into an explicit bijection, but the details are left to future work.  See for example \cite{EdelmanReiner} for some key results on rhombus tilings, and \cite{Orientedmatroids} and for a textbook discussion of oriented matroids.
	\end{rem}
	
	Let us finally express the Parke-Taylor form on the positive configuration space $X^+(3,6)$ as a $d\log$ form in the planar cross-ratios $w_{abc}$:
	\begin{eqnarray}\label{eq: canonical form X36}
		\Omega_{((23456),(12356))} & = & d\log\left(\frac{w^{(1)}_{24}}{w^{(1)}_{35}w^{(1)}_{36}} \right) \wedge d\log\left(\frac{w^{(1)}_{25}w^{(1)}_{35}}{w^{(1)}_{46}} \right) \wedge d\log\left(\frac{w^{(4)}_{15}}{w^{(4)}_{26}w^{(4)}_{36}} \right)\wedge d\log\left(\frac{w^{(4)}_{25}}{w^{(4)}_{26}w^{(4)}_{13}} \right)\nonumber\\
		& = & d\log \left(\frac{w_{124}}{w_{135} w_{136} w_{235} w_{236}}\right)\wedge d\log \left(\frac{w_{125} w_{135} w_{235}}{w_{146} w_{246} w_{346}}\right)\\
		&\wedge  & d\log \left(\frac{w_{145}}{w_{246} w_{256} w_{346} w_{356}}\right)\wedge d\log \left(\frac{w_{245} w_{246} w_{256}}{w_{134} w_{135} w_{136}}\right)\nonumber\\
		& = & d\log\left(\frac{p_{125} p_{134}}{p_{123} p_{145}}\right)\wedge d\log\left(\frac{p_{126} p_{145}}{p_{124} p_{156}}\right)\wedge d\log\left(\frac{p_{146} p_{245}}{p_{124} p_{456}}\right)\wedge d\log\left(\frac{p_{124} p_{345}}{p_{145} p_{234}}\right).\nonumber
	\end{eqnarray}
	We note for posterity that all $14$ planar cross-ratios appear.
	
	It is not difficult to verify using the formula in in Appendix \ref{sec: CEGM integrand calculation} that this form is proportional to the Parke-Taylor factor
	$$PT^{(3)}(1,2,3,4,5,6) =\frac{1}{p_{123} p_{234} p_{345} p_{456} p_{561} p_{612}}.$$
	It is natural to wonder if the Parke-Taylor forms can always be expressed as a single dlog in this way, but it is not the case, as one can see already happens for $X(3,7)$.  In this case, the Parke-Taylor forms can be expressed as a linear combination of (no fewer than) two dlog forms of the form $\Omega_{{\overrightarrow{\sigma}}}$.
	
	Let us continue just a bit more, hewing more closely to \cite{CEZ2022,CEZ2023}.  We following the conventions there, in calling the four types of connected components of $X(3,6)$ types $0,I,II,III$.  
	
	The type I CEGM integrand associated to the GCO
	$$\Sigma=(23456,13465,12654,12653,13426,13425)$$
	is 
	$$\mathcal{I}_\Sigma =\frac{1}{p_{123} p_{134} p_{156} p_{246} p_{256} p_{345}}.$$
	This is proportional to the following linear combination of dlog forms:
	\begin{eqnarray*}
		\Omega_{((23456),(12653))} + \Omega_{((23456),(12643))}.
	\end{eqnarray*}
	We find that $\Omega_{((23456),(12653))}$ is proportional to $$\frac{p_{146}}{p_{123} p_{126} p_{134} p_{156} p_{264} p_{354} p_{465}},$$
	while $ \Omega_{((23456),(12643))}$ is proportional to 
	$$\frac{1}{p_{123} p_{126} p_{134} p_{256} p_{345} p_{456}}.$$
	Now let us argue that a type II CEGM integrand is realized as a linear combination of our dlog forms.  First, $\Omega_{((25436),(12653))}$ is proportional to  
	$$-\frac{1}{p_{125} p_{126} p_{134} p_{236} p_{345} p_{456}}+\frac{p_{346}}{p_{125} p_{134} p_{136} p_{236} p_{246} p_{345} p_{456}} = \frac{p_{146}}{p_{125} p_{126} p_{134} p_{136} p_{246} p_{345} p_{456}},$$
	having taken the sum of the CEGM integrands associated to the pair of GCOs having the same fixed slots 1 and 4,
	$$(25436,15436,12654,12653,12643,12345),\ (25436,13645,14526,12653,12643,13245).$$
	These are GCOs of types I and II, respectively.  Now both types O and I are already dlog forms; from the two-term expression we have constructed it follows that the type II integrand is a sum of two dlog forms.	 
	
	It remains to find an expression for the type III integrand.  A direct computation reveals that a type III is involved in the expression for three GCOs, all of which are already expressed in terms of dlog forms.

	We conclude with an irreducible decoupling identity.  We shuffle $6$ in the second slot; we expect to find the standard decoupling identity on forms,
	\begin{eqnarray}\label{eq: decoupling identity 36}
	\Omega_{((23456),(13456))} - \Omega_{((23456),(13465))} + \Omega_{((23456),(13645))} - \Omega_{((23456),(16345))} = 0,
	\end{eqnarray}
	noting that the minus signs are added for identities among differential forms and not Parke-Taylor factors, as in Equation \eqref{eq: color-dressed form M0n}.
	
	Our strategy is to compute all GCOs such that the first entry is $(23456)$ and the second entry is among the set of cyclic orders
	$$\{(13456),(16345),(13465),(13645)\}.$$
	There are $8,2,3,2$ GCOs that project to these four cyclic orders, respectively.  These 15 GCOs are respectively:
	$$
	\begin{array}{cccccc}
		23456 & 13456 & 12456 & 12356 & 12346 & 12345 \\
		23456 & 13456 & 12654 & 12653 & 12643 & 12543 \\
		23456 & 13456 & 12456 & 12365 & 12364 & 12354 \\
		23456 & 13456 & 12546 & 12536 & 12436 & 12345 \\
		23456 & 13456 & 12645 & 12635 & 12634 & 12543 \\
		23456 & 13456 & 12654 & 12563 & 12463 & 12453 \\
		23456 & 13456 & 12465 & 12365 & 12634 & 12534 \\
		23456 & 13456 & 12564 & 12563 & 12436 & 12435 \\
	\end{array},$$
	with integrands
	\begin{eqnarray*}
	& & \frac{1}{p_{123} p_{126} p_{156} p_{234} p_{345} p_{456}},\frac{1}{p_{123} p_{126} p_{134} p_{256} p_{345} p_{456}},\frac{1}{p_{123} p_{126} p_{145} p_{234} p_{356} p_{456}},\\
	& & \frac{1}{p_{123} p_{126} p_{156} p_{245} p_{345} p_{346}},\frac{1}{p_{123} p_{126} p_{145} p_{256} p_{345} p_{346}},\frac{1}{p_{123} p_{126} p_{134} p_{245} p_{356} p_{456}},\\
	& & \frac{p_{236}}{p_{123} p_{126} p_{145} p_{234} p_{256} p_{346} p_{356}},\frac{p_{136}}{p_{123} p_{126} p_{134} p_{156} p_{245} p_{346} p_{356}},
\end{eqnarray*}
	and 
	$$
	\begin{array}{cccccc}
		23456 & 16345 & 15426 & 15326 & 14326 & 12345 \\
		23456 & 16345 & 14526 & 13526 & 13426 & 12345 \\
	\end{array},$$
	with integrands
	\begin{eqnarray*}
		\frac{1}{p_{126} p_{145} p_{156} p_{234} p_{236} p_{345}},\frac{1}{p_{126} p_{134} p_{156} p_{236} p_{245} p_{345}},
	\end{eqnarray*}
	and 
	$$
	\begin{array}{cccccc}
		23456 & 13465 & 12465 & 12365 & 14326 & 14325 \\
		23456 & 13465 & 12654 & 12653 & 13426 & 13425 \\
		23456 & 13465 & 12645 & 12635 & 14326 & 13425\\
	\end{array},$$
	with integrands
	\begin{eqnarray*}
		\frac{1}{p_{123} p_{145} p_{156} p_{234} p_{256} p_{346}},\frac{1}{p_{123} p_{134} p_{156} p_{246} p_{256} p_{345}},\frac{p_{456}}{p_{123} p_{145} p_{156} p_{246} p_{256} p_{345} p_{346}}
	\end{eqnarray*}
	and
	$$
	\begin{array}{cccccc}
		23456 & 13645 & 12645 & 15326 & 14326 & 13245 \\
		23456 & 13645 & 12654 & 13526 & 13426 & 13245 \\
	\end{array}$$
	with integrands
	\begin{eqnarray*}
		\frac{1}{p_{123} p_{145} p_{156} p_{236} p_{246} p_{345}},\frac{p_{234}}{p_{123} p_{134} p_{156} p_{236} p_{245} p_{246} p_{345}}.
	\end{eqnarray*}
	Here the integrands are given up to sign; we leave it to the reader to find the coefficients $\pm1$ which yield the expressions obtained using differential forms.  Namely: either by summing the integrands in each group or via the formula in Appendix \ref{sec: CEGM integrand calculation}, the claim is that the four differential forms are proportional to
\begin{eqnarray*}
	\Omega_{((23456),(13456))} & \sim & \frac{p_{124} p_{125}}{p_{123} p_{126} p_{134} p_{145} p_{156} p_{234} p_{245} p_{256}}\\
	\Omega_{((23456),(13465))} & \sim & \frac{p_{124}}{p_{126} p_{134} p_{145} p_{165} p_{234} p_{236} p_{245}}\\
	\Omega_{((23456),(13645))} & \sim  & \frac{p_{124}}{p_{123} p_{134} p_{145} p_{156} p_{234} p_{246} p_{265}}\\
	\Omega_{((23456),(16345))} & \sim & \frac{p_{124}}{p_{123} p_{134} p_{145} p_{156} p_{236} p_{245} p_{264}}.
\end{eqnarray*}
On the level of integrands, our decoupling identity becomes
\begin{eqnarray*}
	&& \mathcal{I}_{((23456),(13456))} + \mathcal{I}_{((23456),(16345))}+\mathcal{I}_{((23456),(13465))}+\mathcal{I}_{((23456),(13645))}\\
	& = & \frac{p_{124} p_{125}}{p_{123} p_{126} p_{134} p_{145} p_{156} p_{234} p_{245} p_{256}}+\frac{p_{124}}{p_{126} p_{134} p_{145} p_{165} p_{234} p_{236} p_{245}}+\frac{p_{124}}{p_{123} p_{134} p_{145} p_{156} p_{234} p_{246} p_{265}}\\
	& + & \frac{p_{124}}{p_{123} p_{134} p_{145} p_{156} p_{236} p_{245} p_{264}}\\
	& = & 0.
\end{eqnarray*}
	This provides a quite nontrivial consistency check between our methods and the results of \cite{CEZ2022,CEZ2023}!

\section{Permutation Invariant Color-Dressed CEGM Amplitudes}\label{sec: permutation invariance}	
In this work, we have constructed color-dressed dlog forms, leading to color-dressed integrands and amplitudes for generalized biadjoint scalar amplitudes, which independently satisfy the $n$ $U(1)$ decoupling relations; however they are not permutation invariant.  However, one can construct, via the CEGM formula, more general color-dressed amplitudes superpositions of the integrands obtained from the dlog form as in Appendix \ref{sec: CEGM integrand calculation}.

We discuss the extension of this construction to fully permutation invariant color-dressed generalized biadjoint scalar amplitudes.  The novel aspect is that for sake of genericity, we replace the usual basis elements  $T^{a_1},\ldots, T^{1_n}$ of the adjoint representation with $(k-1)$-fold tensor products.  We discuss motivation in Remark \ref{rem: motivation color algebra}. 

For each $(k,n)$ with $(2\le k\le n-2)$, denote by $\mathfrak{su}(N)^{\otimes (k-1)}$ the $(k-1)$-fold tensor power $\mathfrak{su}(N)^{\otimes (k-1)} = \mathfrak{su}(N) \otimes \cdots \otimes \mathfrak{su}(N)$ of the adjoint representation of $\mathfrak{su}(N)$.

The color data is stored in a $(k-1)\times n$ matrix $(T^{a_{i,j}})$ with entries given generators $T^{a_{i,j}} \in \mathfrak{su}(N)$ of $SU(N)$ for $i=1,\ldots, k-1$ and $j=1,\ldots, n-(k-2)$.  From this data we construct elements of $\mathfrak{su}(N)^{\otimes (k-1)}$, tensor products of the form
\begin{eqnarray}\label{eq: Ta tensor}
	T^\mathbf{a} & = & T^{a_{1,j_1}}\otimes \cdots \otimes T^{a_{k-1,j_{k-1}}},
\end{eqnarray}
where $(j_1,\ldots, j_{k-1})$ is a tuple (of particle labels) in $\{1,\ldots, n\}$ which are not necessarily distinct.
Color factors $c_{\overrightarrow{\sigma}}$ are traces of certain products of matrices $T^{\mathbf{a}}$ which must simplify to a given product of $k-1$ standard color factors, as specified by $\overrightarrow{\sigma}$.  
To save notation, we will elaborate in an example: we will construct a $(4,8)$ color factor.  Fix $L = (3,2,7)$ and a color order 
$$\overrightarrow{\sigma} = ((154368),(124568),(145678))$$
compatible with $L$.  The matrix $(T^{a_{i,j}})$ is
$$(T^{a_{i,j}})=\left(
\begin{array}{cccccccc}
	T^{a_{1,1}} & \text{\sout{$T^{a_{1,2}}$}}  &T^{a_{1,3}} & T^{a_{1,4}} & T^{a_{1,5}} & T^{a_{1,6}} & \text{\sout{$T^{a_{1,7}}$}} & T^{a_{1,8}} \\
	T^{a_{2,1}} & T^{a_{2,2}} & \text{\sout{$T^{a_{2,3}}$}} & T^{a_{2,4}} & T^{a_{2,5}} & T^{a_{2,6}} & \text{\sout{$T^{a_{2,7}}$}} & T^{a_{2,8}} \\
	T^{a_{3,1}} & \text{\sout{$T^{a_{3,2}}$}} & \text{\sout{$T^{a_{3,3}}$}} & T^{a_{3,4}} & T^{a_{3,5}} & T^{a_{3,6}} & T^{a_{3,7}} & T^{a_{3,8}} \\
\end{array}
\right),$$
where the generators missing from each of the three cyclic orders have been struck out for emphasis.  From 
$\overrightarrow{\sigma}$ we define the following color factor:
$$c_{\overrightarrow{\sigma}} = \text{tr}\left(T^{a_{1,1}}T^{a_{1,5}} T^{a_{1,4}} \cdots T^{a_{1,8}}\right)\text{tr}\left(T^{a_{2,1}}T^{a_{2,2}} T^{a_{2,4}} \cdots T^{a_{2,8}}\right)\text{tr}\left(T^{a_{3,1}}T^{a_{3,4}} T^{a_{3,5}} \cdots T^{a_{3,8}}\right).$$
Now for the main point: by taking advantage of the multiplicativity of the trace on tensor products we can take a ``transpose'' and write this (non-uniquely) in a very pretty way as 
$$c_{\overrightarrow{\sigma}} = \text{tr}\left(\left(T^{a_{1,1}}\otimes T^{a_{2,1}}\otimes T^{a_{3,1}}\right)\left(T^{a_{1,5}}\otimes T^{a_{2,2}}\otimes T^{a_{3,4}}\right)\cdots \left(T^{a_{1,8}}\otimes T^{a_{2,8}}\otimes T^{a_{3,8}}\right)\right),$$
where each $T^{(a_{1,j_1},a_{2,j_2},a_{3,j_3})} := T^{a_{1,j_1}} \otimes T^{a_{2,j_2}} \otimes T^{a_{3,j_3}}$ is in $\mathfrak{su}(N)\otimes \mathfrak{su}(N) \otimes\mathfrak{su}(N)$.  So we see that $(k,n)$ color factors are certain conjugation invariant functions on the tensor product of adjoint representations $\mathfrak{su}(N)^{\otimes(k-1)} = \mathfrak{su}(N) \otimes \cdots \otimes \mathfrak{su}(N)$.

With this explicit realization of the color factors in hand, we are ready to define the full permutation invariant color-dressed amplitude.  

The formula in Appendix \ref{sec: CEGM integrand calculation} allows us to compute color-dressed differential forms $\Omega^{(L)}_\mathbf{c}$, and we write simply $\mathcal{I}^{(L)}_\mathbf{c}$ for the color-dressed integrand obtained by applying Equation \eqref{eq: integrand from dlog form} to each dlog form.

More concretely, we write $\mathcal{I}^{L}_\mathbf{c}:=\mathcal{I}\left(\Omega^{L}_\mathbf{c}\right)$, or  for the color-dressed integrand obtained from the differential form 
$$\Omega^{L}_\mathbf{c} = \Omega^{(L\setminus \{\ell_1\})}_{c_1} \wedge \cdots \wedge \Omega^{(L\setminus \{\ell_{k-1}\})}_{c_{k-1}}$$
using Appendix \ref{sec: CEGM integrand calculation}.  Here the notation $\Omega^{(L\setminus \{\ell_j\})}_{c_j}$ means that its color factors are traces of products of elements of the $j^\text{th}$ copy of $\mathfrak{su}(N)$ in $\mathfrak{su}(N)^{\otimes (k-1)}$.

\begin{defn}\label{defn: fully symmetric CEGM amplitude}
	The permutation invariant color-dressed integrand $\mathcal{I}_\mathbf{c}$ is given by 
	\begin{eqnarray}\label{eq: color dressed integrand generator perm invariant}
		\mathcal{I}_\mathbf{c} & = & \sum_{\overrightarrow{\sigma} \in \text{CO}^{(k)}_n}c_{\overrightarrow{\sigma}} \mathcal{I}_{\overrightarrow{\sigma}}
	\end{eqnarray}
	where the sum is over the set $\text{CO}^{(k)}_n$ of all type $(k,n)$-color orders $\overrightarrow{\sigma} = (\sigma^1,\ldots,\sigma^{k-1})$.  Here
	$$c_{\overrightarrow{\sigma}} = \prod_{j=1}^{k-1}\text{tr}\left(T^{a_{j,\sigma^j_1}}T^{a_{j,\sigma^j_2}}\cdots T^{a_{j,\sigma^j_{n-(k-2)}}}\right),$$
	and $\mathcal{I}_{\overrightarrow{\sigma}}$ is the integrand computed using Appendix \ref{sec: CEGM integrand calculation},
	$$\mathcal{I}_{\overrightarrow{\sigma}} = \mathcal{I}\left(\Omega^{L\setminus\{\ell_1\}}_{\sigma^1}\wedge \cdots \wedge \Omega^{L\setminus\{\ell_{k-1}\}}_{\sigma^{k-1}}\right).$$
\end{defn}

\vspace{.1in}

Let us now formulate in more detail the case $(k,n) = (3,n)$, where the notation can be slightly simplified.  We fix two tuples of generators of $SU(N)$,
$$T^\mathbf{a} = (T^{a_1},\ldots, T^{a_{n}}),\ \ T^\mathbf{b} =  (T^{b_1},\ldots, T^{b_{n}})$$
and from these construct $n^2$ elements $T^{(a_i,b_j)} = T^{a_i} \otimes T^{b_j} \in \mathfrak{su}_N \otimes \mathfrak{su}_N$.

Given, as usual, any $L = (\ell_1,\ell_2)$ and any $(3,n)$ color order $\overrightarrow{\sigma} = ((\sigma^1_{i_1},\ldots, \sigma^1_{i_{n-1}}),(\sigma^2_{j_1},\ldots, \sigma^2_{j_{n-1}}) )$ compatible with $L$, then due to the multiplicativity of traces for the tensor product, we have
\begin{eqnarray}\label{eqn: color factors (3,n)}
	c_{\overrightarrow{\sigma}} & = & \text{tr}\left(T^{a_{\sigma^1_{i_1}}}\cdots T^{a_{\sigma^1_{i_{n-1}}}}\right)\text{tr}\left(T^{b_{\sigma^2_{j_1}}}\cdots T^{b_{\sigma^2_{j_{n-1}}}}\right),
\end{eqnarray}
where we denote by $(i_1,\ldots, i_{n-1})$ and $(j_1,\ldots, j_{n-1}) $ the cyclic orders obtained from $(1,2,\ldots, n)$ by deleting $\ell_1,\ell_2$, respectively.

We observe that due to the double cyclic symmetry of the color factor $c_{\overrightarrow{\sigma}}$, different collections of $n-1$ generators $T^{(a_i,b_j)}$ can give rise to the same color factor $\mathbf{c}_{(\sigma^1,\sigma^2)}$, so long as the cyclic orders $\sigma^1$ and $\sigma^2$ remain the same.  According to Definition \ref{defn: fully symmetric CEGM amplitude}, the permutation invariant color-dressed integrand is written as 
	\begin{eqnarray}
	\mathcal{I}_\mathbf{c} & = & \sum_{\overrightarrow{\sigma} \in \text{CO}^{(3)}_n}c_{\overrightarrow{\sigma}} \mathcal{I}_{\overrightarrow{\sigma}}.
\end{eqnarray}

We claim that with respect to $\mathfrak{su}(N)^{\otimes (2)}$, partial amplitudes of the permutation invariant color-dressed amplitude are exactly CEGM integrals of products $\mathcal{I}_{\overrightarrow{\sigma_L}}\mathcal{I}_{\overrightarrow{\sigma_R}}$ for pairs of (3,n) color orders $\overrightarrow{\sigma_L},\overrightarrow{\sigma_R}$.

\begin{rem}\label{rem: motivation color algebra}
	Let us motivate our choice of the color algebra: why, for the fully permutation invariant color-dressed amplitude, should we use the $(k-1)$-fold tensor product of adjoint representations $\mathfrak{su}(N)^{\otimes (k-1)}$ rather than a single $\mathfrak{su}(N)$?  We explain in two special cases for $(k,n) = (3,5),(3,6)$, why for sake of genericity we need $k-1 = 2$ sets of $n$ generators for $U(N)$.  We may follow the notation in Equation \ref{eqn: color factors (3,n)}, writing 
	$$c_{((i_1i_2i_3i_4),(j_1j_2j_3j_4))} = \text{tr}\left(T^{a_{i_1}}T^{a_{i_2}}T^{a_{i_3}}T^{a_{i_4}}\right)\text{tr}\left(T^{b_{j_1}}T^{b_{j_2}}T^{b_{j_3}}T^{b_{j_4}}\right).$$

	After extracting the function as in Appendix \ref{sec: CEGM integrand calculation}, we find that
	$$\mathcal{I}_{((2345),(1345))} = \frac{p_{124}}{p_{123}p_{134}p_{145}p_{234}p_{245}p_{125}}.$$
	By symmetrizing 
	$$c_{((2345),(1345))}\frac{p_{124}}{p_{123}p_{134}p_{145}p_{234}p_{245}p_{125}}$$
	with the (unsigned) stabilizer group of 
	$$\frac{p_{124}}{p_{123}p_{134}p_{145}p_{234}p_{245}p_{125}},$$
	the numerical coefficient of 
	$$\frac{p_{124}}{p_{123}p_{134}p_{145}p_{234}p_{245}p_{125}}$$
	in the sum is the linear combination of 12 color factors
		\begin{eqnarray}
	& & c_{((1345),(1325))} - c_{((2543),(1325))} - c_{((1325),(1345))} + c_{((2345),(1345))} + c_{((1543),(1523))} - c_{((2345),(1523))}\\
	& - &c_{((1523),(1543))} + c_{((2543),(1543))} - c_{((1345),(2345))}-c_{((1523),(2345))} + c_{((1325),(2543))} - c_{((1543),(2543))} ,\nonumber
\end{eqnarray}
	which is antisymmetric under the exchange $(\sigma^1,\sigma^2) \Leftrightarrow(\sigma^2,\sigma^1)$ and thus would vanish identically if $a_i=b_i$ for all $i=1,\ldots, 5$.

	For sake of contrast, let us consider a possible integrand for $(k,n) = (3,6)$.  Repeating the analogous calculation as above on
	$$c_{((23456),(13456))}\frac{p_{124}p_{125}}{p_{123}p_{134}p_{145}p_{156}p_{234}p_{245}p_{256}p_{126}},$$
	we find a linear combination of four color factors
	$$c_{((23456),(13456))} + c_{((26543),(16543))} + c_{((13456),(23456))} + c_{((16543),(26543))},$$
	which is \textit{symmetric} under the exchange $a_i\leftrightarrow b_i$.
	
\end{rem}

Using $\mathfrak{su}(N)^{\otimes (k-1)}$ to construct color factors appears to contain some surprises, as might expect (and hope!) and we have only just barely initiated the investigation.  For example, by the duality $X(2,n)\simeq X(n-2,n)$ we find a decomposition of the $n$-particle color-dressed biadjoint scalar into partial amplitudes which is much finer than its usual decomposition into partial amplitudes.  We have also encountered certain partial decoupling identities for the permutation invariant color-dressed amplitudes, but they remain quite mysterious, which begs for a deeper exploration.  For $(k,n) = (3,5)$ there are now two distinct isomorphism types of residue posets of partial amplitudes $m^{(3)}(\overrightarrow{\sigma},\overrightarrow{\sigma})$: not only the usual 2-dimensional associahedra, that is pentagons, but also hexagons.  What are the possible residue posets of partial amplitudes $m^{(k)}(\overrightarrow{\sigma},\overrightarrow{\sigma})$ when $n=k+2$ for general $k$?  When the color order consists of $k-1$ cyclic orders that are obtained by deleting $k-2$ labels from a global cyclic order $(1,2,\ldots, k+2)$, we infer that (the compactification of) the tile in $X(k-2,k)$ should be a curvy analog of the $(k-1)$-dimensional standard permutohedron.  What other possible polytopes do we encounter as the color order $\overrightarrow{\sigma}$ varies over all admissible $(k-1)$-tuples of cyclic orders?

We leave these fascinating questions to future work, except to conclude with a related, very concrete question: can one give a general prescription for the coefficients of Generalized Feynman Diagrams for permutation invariant color-dressed amplitudes?  One could also ask questions in the context of color-kinematics duality as discussed in \cite[Section 8]{ABHY2018}.

	\section{Acknowledgements}
	We thank Freddy Cachazo for encouragement to write this paper and for many helpful discussions and comments on the draft.  We also thank Nima Arkani-Hamed, Thomas Lam, Alexander Tumanov and Yong Zhang for explanations about canonical forms, spirited discussions and fruitful collaboration.  We are grateful to Perimeter Institute for Theoretical Physics for excellent working conditions while this paper was written.

\newpage
	\appendix

	\section{Planar Cross-Ratios}\label{sec:planar cross-ratios}
	Here we provide Mathematica code for planar cross-ratios $w_J$ and $w^{(L)}_{i,j}$, respectively.
	\begin{verbatim}
		crossRatioKNProd[n_][J_]:=Module[{k,binary,cube,out},
		k=Length[J];
		binary=Total[Map[UnitVector[n,#]&,J]];
		cube=Map[#+binary&,Map[Total,Subsets[Select[
		Permute[-(UnitVector[n,1]-UnitVector[n,2]),CyclicGroup[n]],
		Union[#+binary]==={0,1}&]]]];
		out=Times@@Map[(p[Flatten[Position[#,1]]])^((-1)^(k-1+Length[
		Intersection[Flatten[Position[#,1]],J]]))&,cube];
		out
		]
	\end{verbatim}
	
	\begin{verbatim}
		crossRatioOctah[{k_,n_}][{L_,{i_,j_}}]:=
		If[Or[MemberQ[Map[Sort,PermutationReplace[Range[k],CyclicGroup[n]]],
		Union[Join[L,{i,j}]]],k=!=Length[Union[Join[L,{i,j}]]],
		Sort[Join[L,{i,j}]]=!=Union[Join[L,{i,j}]]],
		{},
		Module[{range,pos,out,pluckers},
		range=Complement[Range[n],L];
		pos={Position[range,i],Position[range,j]};
		pluckers=Map[p,Table[Sort[Join[L,range[[Flatten[Mod[{pos[[1]],pos[[2]]}+X,
		n-(k-2),1]]]]]],{X,{{0,0},{1,1},{1,0},{0,1}}}]];
		(pluckers[[3]]pluckers[[4]])/(pluckers[[1]]pluckers[[2]])
		]
		]
	\end{verbatim}
	
\section{Integrands from dlog Forms}\label{sec: CEGM integrand calculation}
For sake of completeness we give a standard procedure to produce a CEGM integrand from the dlog expression which is our starting point.  Given a wedge product of the form
$$\Omega_{\overrightarrow{\sigma}} = \bigwedge_{i=1}^{k-1}\Omega_{{\overrightarrow{\sigma}}_i},$$
where each $\Omega_{\sigma_i}$ is the pullback of the canonical form on some connected component of $M_{0,n-(k-2)}$, write 
	$$M(x) = \begin{bmatrix}
	x_{1,1} & x_{1,2} & \cdots  & x_{1,n} \\
	x_{2,1} & x_{2,2} &  & x_{2,n} \\
	\vdots  &  & \ddots & \vdots  \\
	x_{k-1,1} & x_{k-1,2} & \cdots & x_{k-1,n-k-1}\\
	1 & 1 & \cdots & 1
\end{bmatrix}.$$
	It is now easy to recognize from the Jacobian the expression in terms of Plucker variables.
	Let us fix a subset of $k+1$ columns, say $I = \{1,2,\ldots, k+1\}$; we use $I$ as our projective frame and treat the variables in the columns $k+2,\ldots, n$ of $M$ as coordinates.  Now, note that 
$$\Omega_{{\overrightarrow{\sigma}}} = d\log(w_1)\wedge \cdots \wedge d\log(w_{(k-1)(n-k-1)})$$
for some (given) cross-ratios $w_1,\ldots, w_{(k-1)(n-k-1)}$.
	
Fix an order on the coordinates $(z_1,\ldots, z_{(k-1)(n-k-1)})$ such that $\{z_1,\ldots, z_{(k-1)(n-k-1)}\} = \{x_{i,j}: (i,j) \in \lbrack 1,k-1\rbrack \times \lbrack k+2,n\rbrack\}$
	
	We define $\mathcal{I}(\Omega_{{\overrightarrow{\sigma}}})$, or for short just $\mathcal{I}_{\overrightarrow{\sigma}}$, by 
	\begin{eqnarray}\label{eq: integrand from dlog form}
		\mathcal{I}_{{\overrightarrow{\sigma}}} = \frac{\mathbf{J}}{\prod_{j=1}^{k+1}p_{\lbrack 1,k+1\rbrack\setminus \{j\}}},
	\end{eqnarray}
	where $\mathbf{J}$ is the Jacobian determinant
	$$\mathbf{J} = \frac{\partial(\log(w_1),\ldots, \log(w_{(k-1)(n-k-1)}))}{\partial(z_1,\ldots, z_{(k-1)(n-k-1)})}$$
	We claim that $\mathcal{I}_{{\overrightarrow{\sigma}}}$ is independent of the subset $I$.
	
	This allows us to construct integrands (which we recognize as functions of the Plucker variables) from our dlog forms.  In particular, this formula applies to the color-dressed differential forms $\Omega^{(L)}_\mathbf{c}$, and for the integrand we write simply $\mathcal{I}_\mathbf{c}$.

\end{document}